%% file: ms.tex
\documentclass[journal]{IEEEtran}

\pdfoutput=1

\usepackage{etoolbox}
\newbool{isgitdraft}
\boolfalse{isgitdraft}

\usepackage{xcolor}
\usepackage[normalem]{ulem}
\newbool{isReviewUpdate}
\boolfalse{isReviewUpdate}

\usepackage{cite}

\usepackage{multirow}

\usepackage[pdftex]{graphicx}
\graphicspath{{imgs/}}
\DeclareGraphicsExtensions{.pdf,.jpeg,.png}

\usepackage{amsmath}
\usepackage{amsfonts}
\usepackage{amsthm}
\usepackage{amssymb}
\usepackage{mathtools}

\usepackage{algpseudocode}
\usepackage{algorithm}

\algnewcommand{\LeftComment}[1]{\Statex \(\triangleright\) #1}

\usepackage{array}

\usepackage[caption=false,font=footnotesize]{subfig}

\usepackage{url}

\newtheorem{theorem}{\textbf{Theorem}}[section]

\newtheorem{lemma}{\textbf{Lemma}}[section]

\newtheorem{remark}{\textit{Remark}}

\ifbool{isgitdraft}{
  \usepackage[mark]{gitinfo2}
}

\usepackage{acro}
\input{sections/acronyms}

\begin{document}

\ifbool{isgitdraft}{
  \begin{titlepage}
    \centering
    \vspace*{5cm}
    \huge\textit{\textbf{DRAFT}}: Multi-sensor Joint Adaptive Birth Sampler for Labeled Random Finite Set Tracking\\
    \vspace{2\baselineskip}
    \large Compiled on \today\\
    \vspace{2\baselineskip}
    \textbf{Git Hash}: \texttt{\gitHash}\\
    \vspace{2\baselineskip}
    \textbf{Git Branch}: \texttt{\gitBranch}\\
    \vfill
  \end{titlepage}
  \newpage
}

\title{Multi-sensor Joint Adaptive Birth Sampler for Labeled Random Finite Set Tracking}

\author{
    Anthony~Trezza,~\IEEEmembership{Member,~IEEE,}
    Donald~J.~Bucci~Jr.,~\IEEEmembership{Senior Member,~IEEE,}
    Pramod~K.~Varshney,~\IEEEmembership{Life Fellow,~IEEE}
    \thanks{A. Trezza and D.J. Bucci are with Lockheed Martin Advanced Technology Labs, Cherry Hill, NJ 08002 USA (e-mail: anthony.t.trezza@lmco.com; donald.j.bucci.jr@lmco.com).}
    \thanks{P. K. Varshney is with the Department of Electrical Engineering and Computer Science, Syracuse University, Syracuse, NY 13244 USA (e-mail: varshney@syr.edu).}
}

\maketitle
\input{sections/abstract}
\IEEEpeerreviewmaketitle
\input{sections/introduction}
\input{sections/background}

\input{sections/problem_formulation}
\input{sections/gibbs_sampler}
\input{sections/monte_carlo.tex}
\input{sections/gaussian.tex}
\input{sections/simulations.tex}
\input{sections/conclusions.tex}
\appendices
\input{sections/appendix_technical_lemmas.tex}
\input{sections/appendix_l1_dist.tex}
\input{sections/appendix_gaussian_expval.tex}
\input{sections/appendix_gaussian_sampling.tex}
\input{sections/appendix_gaussian_spatial.tex}
\bibliographystyle{IEEEtran}
\bibliography{IEEEabrv, sections/labeled_msjab.bib}

%

%
%
%

\end{document}

%% file: sections/acronyms.tex
\DeclareAcronym{RFS}{
	short = RFS,
	long = Random Finite Sets
}

\DeclareAcronym{PMBM}{
	short = PMBM,
	long = Poisson Multi-Bernoulli Mixture
}

\DeclareAcronym{LMB}{
	short = LMB,
	long = Labeled Multi-Bernoulli
}
\DeclareAcronym{GLMB}{
	short = GLMB,
	long = Generalized Labeled Multi-Bernoulli
}

\DeclareAcronym{JPDA}{
	short = JPDA,
	long = Joint Probabilistic Data Association
}

\DeclareAcronym{MHT}{
	short = MHT,
	long = Multiple Hypothesis Tracking
}

\DeclareAcronym{BP}{
	short = BP,
	long = Belief Propagation
}

\DeclareAcronym{GNN}{
	short = GNN,
	long = Global Nearest Neighbor
}

\DeclareAcronym{FISST}{
	short = FISST,
	long = Finite Set Statistics
}

\DeclareAcronym{PHD}{
	short = PHD,
	long = Probability Hypothesis Density
}

\DeclareAcronym{CPHD}{
	short = CPHD,
	long = Cardinalized Probability Hypothesis Density
}

\DeclareAcronym{CBMEMBER}{
	short = CB-MeMBeR,
	long = Cardinality Balanced Multi-target Multi-Bernoulli
}

\DeclareAcronym{TDOA}{
	short = TDOA,
	long = time difference of arrival
}

\DeclareAcronym{FDOA}{
	short = FDOA,
	long = frequency difference of arrival
}

\DeclareAcronym{OSPA}{
	short = OSPA,
	long = Optimal Subpattern Assignment
}

\DeclareAcronym{GCI}{
	short = GCI,
	long = Generalized Covariance Intersection
}

%% file: sections/abstract.tex
%
\begin{abstract}
This paper provides a scalable, multi-sensor measurement adaptive track initiation technique for labeled random finite set filters.
A naive construction of the multi-sensor measurement adaptive birth set distribution leads to an exponential number of newborn components in the number of sensors.
A truncation criterion is established for a labeled multi-Bernoulli random finite set birth density.
The proposed truncation criterion is shown to have a bounded L1 error in the generalized labeled multi-Bernoulli posterior density.
This criterion is used to construct a Gibbs sampler that produces a truncated measurement-generated labeled multi-Bernoulli birth distribution with quadratic complexity in the number of sensors.
A closed-form solution of the conditional sampling distribution assuming linear Gaussian likelihoods is provided, alongside an approximate solution using Monte Carlo importance sampling.
Multiple simulation results are provided to verify the efficacy of the truncation criterion, as well as the reduction in complexity.
\end{abstract}


%% file: sections/introduction.tex
%
\section{Introduction}\label{sec::introduction}
\IEEEPARstart{T}{he} goal of a multi-object tracking algorithm is to estimate the number of objects and their trajectories from measurements observed at one or more sensors.
Many tracking approaches have surfaced including, \ac{GNN} techniques~\cite{Blackman1999},~\ac{JPDA}~\cite{Barshalom2009},~\ac{MHT}~\cite{Blackman2004},~\ac{BP}~\cite{Meyer2018}, and~\ac{RFS}~\cite{Mahler2007, Mahler2014}.
We direct the reader to~\cite{Vo2015} for a detailed survey of the field, recent advances, and example applications.
In this work, we focus on the canonical problem formulation where point objects are observed by multiple sensors at discrete time instants and are incorporated into state estimates via an online filtering recursion.

A key component of multi-object tracking techniques involves constructing newborn object tracks (i.e., \emph{track initialization}).
In~\ac{GNN},~\ac{MHT} and~\ac{JPDA} techniques, newborn objects are constructed from unassociated measurements using application specific procedures~\cite{Kennedy2008, Hu1997}.
In contrast,~\ac{RFS} trackers leverage concepts from~\ac{FISST}~\cite{Vo2013, Mahler2007} to create \emph{multi-object} prior distributions representing newborn objects.
For the~\ac{PHD}~\cite{Vo2006},~\ac{CPHD}~\cite{Vo2007} and~\ac{PMBM}~\cite{Garcia2018} filters, the multi-object prior is a Poisson~\ac{RFS} describing the average number of newborn objects and their joint spatial distribution.
For the~\ac{LMB}~\cite{Reuter2014, Reuter2017} and~\ac{GLMB}~\cite{Vo2014, Vo2017, Vo2019} filters, the multi-object prior is a~\ac{LMB}~\ac{RFS} representing tuples of birth probabilities and spatial distributions for each newborn object.
In a \emph{static birth} strategy, the multi-object prior is fixed for the duration of a filter's runtime and encodes known prior information.
Static birth strategies are typically used when objects enter the surveillance volume in known predictable locations (e.g., air traffic control~\cite{Skolnik2008}).
However, they do not include methods for re-acquiring dropped tracks.
In a \emph{measurement adaptive birth} strategy, the multi-object prior is determined from measurements each time the filtering recursion is called.
This approach is effective in many applications since minimal prior information is known about where and how objects can appear in the surveillance volume.
Care must be taken when designing these strategies to ensure that tracker performance is maintained without adversely affecting computational complexity.

For the remainder of this paper, we will limit our discussion to object birthing strategies for~\ac{RFS} track filtering algorithms.
Measurement adaptive birth strategies for~\ac{RFS} track filtering algorithms have been discussed extensively in \emph{single-sensor} applications.
These strategies were first formalized in the RFS tracking literature for the~\ac{PHD} and~\ac{CPHD} filters in~\cite{Ristic2012}.
The authors proposed an augmented state space model, allowing for a separate specification of the~\ac{PHD} intensity function for newborn objects.
They then focused on efficient particle placement strategies under a sequential Monte Carlo (i.e., particle) realization of the~\ac{PHD} and~\ac{CPHD} filters.
This technique was later generalized for observable and unobservable state space partitions (i.e., partially uniform) under Gaussian Mixture belief states in~\cite{Beard2013}.
Single-sensor adaptive birth techniques for the~\ac{CBMEMBER} filter were proposed in~\cite{Reuter2013, Changshun2018, Hu2019}.

Single-sensor adaptive birth strategies were introduced in~\cite{Reuter2014} for the~\ac{LMB} filter and in~\cite{Lin2016} for the~\ac{GLMB} filter.
These techniques dynamically construct a \ac{LMB} multi-object prior distribution using each measurement.
The birth rate is controlled by subdividing a fixed birth rate proportionally per \ac{LMB} component based on the relative probabilities of associating with any persisting object.
The spatial distributions per newborn target are then constructed using the approach in~\cite{Ristic2012}.
Generalizations of these techniques have recently been proposed in~\cite{Legrand2018, Zhu2021, Hoher2020,Do2020,Yoon2011,Liu2020}.
The authors in~\cite{Legrand2018} note that the birth technique of~\cite{Reuter2014, Lin2016} can lead to multiple targets being born from the same measurement.
To address this, they propose an alternative~\ac{GLMB} filter structure that models the birth distribution using a labeled Poisson~\ac{RFS} and shifts the target birth procedure from the prediction to the update step.
The authors in~\cite{Zhu2021} propose an adaptive birth model for the particle~\ac{GLMB} filter based on interval measurements and relevance likelihood functions.
A detection-driven approach based on Rauch-Tung-Striebel (RTS) smoothing to adaptively refine the birth distribution of the LMB filter is provided in~\cite{Hoher2020}.
Another approach based on running a parallel~\ac{CPHD} filter to bootstrap the~\ac{GLMB} filter at each time step is proposed in~\cite{Do2020}.
Multi-time step initialization techniques for unlabeled and labeled~\ac{RFS} filtering algorithms are provided in~\cite{Yoon2011} and~\cite{Liu2020} respectively.

In contrast to the single-sensor adaptive birth strategies, multi-sensor multi-target filtering studies typically assume a static birth distribution (e.g.,~\cite{Gostar2021, Tian2019, Vo2019, Meyer2018, Frohle2019}).
Multi-sensor multi-object adaptive birth strategies have not yet been investigated systematically in the~\ac{RFS} tracking literature, especially for labeled~\ac{RFS} filters.
The authors in~\cite{Lanterman2008} propose a clustering method for newborn objects for a particle~\ac{PHD} filter via a coarse discretization of bistatic range measurements, projected into the object state space. 
A multi-sensor adaptive birth approach is proposed for the~\ac{PHD} filter in~\cite{Berry2018} based on the iterated-corrector heuristic and applied to a tracking problem using~\ac{TDOA} and~\ac{FDOA} measurements.
A brute force extension of the~\ac{GLMB} filter adaptive birth technique proposed in~\cite{Lin2016} is suggested in~\cite{Liu2017} for non-overlapping fields of view.

Alternatively, the multi-sensor multi-target filtering problem can be formulated using a track-to-track fusion architecture.
Centralized fusion architectures assume that a centralized fusion node has access to every sensor's measurements which are jointly used to refine the global multi-object posterior density.
Track-to-track fusion architectures are structured such that each sensing node runs a local single-sensor multi-object tracking algorithm, using its own measurements to refine a local multi-object posterior.
Then, each sensing node's local multi-object posterior is shared and fused either at a centralized fusion node or using a distributed processing technique to obtain an approximation of the global multi-object posterior.
In these approaches, each sensor's local multi-object tracker is responsible for constructing a newborn multi-object density at each time step.
When fusing labeled~\ac{RFS} multi-object posteriors, the track-to-track fusion procedure becomes more challenging due to the possibility of label inconsistency (i.e., the same object may have been initialized with different labels in each local sensor's filter).
This occurs when each sensor's local multi-object tracker has a different static multi-object birth prior or when they use a single-sensor adaptive birth procedure.
Several centralized and distributed multi-sensor multi-object track-to-track fusion architectures have been proposed for labeled~\ac{RFS} filters to address these challenges using approximations such as label matching \cite{Li2019} or local marginalization to unlabeled filter variants \cite{Li2018}.

The lack of a well-defined, systematic approach to multi-sensor multi-object adaptive birth strategies is problematic for many applications. 
In tracking applications, managing the number of ghost tracks is paramount to maintaining tracker accuracy and minimizing runtime complexity~\cite{Bishop2007}.
An exhaustive evaluation of potential object birth distributions from all sensor measurement tuples is exponential in the number of sensors.
Managing this complexity necessitates elimination of implausible births from sensor tuples that do not geometrically cluster~\cite{Jia2017}, or that have incompatible auxiliary features~\cite{Reed2008, Alexandridis2015, Wang2020}.

In this paper, we derive an efficient multi-sensor multi-object adaptive birth strategy for labeled~\ac{RFS} filters using a Gibbs sampler.
We show that this significantly improves the scalability of constructing the adaptive birth~\ac{LMB} multi-object prior by reducing the complexity from exponential to quadratic in the number of sensors without adversely affecting tracking performance.
As opposed to ad hoc clustering techniques~\cite{Jia2017, Berry2018}, our approach exponentially converges to an established truncation criterion that results in a bounded L1 error in the \ac{GLMB} posterior.
The proposed approach can easily be modified to incorporate information from auxiliary features~\cite{Reed2008, Alexandridis2015, Wang2020}, but it is not dependent upon it.
This results in a more robust solution for applications where feature data is unavailable or unreliable.

The major contributions of this paper are,
\begin{itemize}
\item A formal definition of the multi-sensor adaptive birth \ac{LMB} density in the context of~\cite{Reuter2014, Lin2016}.
\item A truncation criterion for a \ac{LMB} birth density and its bounded L1 error in the posterior $\delta$-\ac{GLMB} density.
\item An approach for generating a truncated multi-sensor adaptive birth \ac{LMB} density via Gibbs sampling that achieves quadratic complexity in the number of sensors.
\item A Monte Carlo approximation of the Gibbs sampling algorithm and construction of the multi-sensor adaptive birth \ac{LMB} density.
\item The derivations for the closed-form solution of the Gibbs sampling algorithm and construction of the multi-sensor adaptive birth \ac{LMB} density assuming linear Gaussian models.
\end{itemize}

The paper is organized as follows.
Section~\ref{sec::background} provides background material on labeled~\ac{RFS} tracking and the multi-sensor $\delta$-\ac{GLMB} filtering recursion.
Section~\ref{sec::truncation} presents a truncation criterion for a \ac{LMB} birth density and its bounded L1 error in the posterior $\delta$-\ac{GLMB} density.
Section~\ref{sec::problem} formalizes the multi-sensor adaptive birth problem and its complexity challenges.
Section~\ref{sec::gibbs} derives the proposed Gibbs sampling truncation technique to generate a birth \ac{LMB} density from multi-sensor measurement tuples.
Section~\ref{sec::mc} and Section~\ref{sec::gm} provide the Monte Carlo approximation and the derivation of the closed-form solution under Gaussian likelihoods respectively.
Section~\ref{sec::sim} provides simulated results of the proposed adaptive birth technique used in the~\ac{LMB} and~$\delta$-\ac{GLMB} filters.
Finally, concluding remarks are provided in Section~\ref{sec::conclusions}.

%% file: sections/background.tex
\section{Background}\label{sec::background}

We adopt the following notation from~\cite{Vo2013, Vo2019}.
Single-object states are represented by lowercase letters, e.g. $x, \textbf{x}$.
Multi-object states are represented as uppercase letters, e.g. $X, \textbf{X}$.
Labeled states and their distributions will be represented by bold letters, e.g. $\textbf{x}, \textbf{X}, \boldsymbol{\pi}$.
Spaces will be represented using blackboard bold letters, e.g. $\mathbb{X}, \mathbb{Z}, \mathbb{L}, \mathbb{R}$, etc.
The sequence of variables $X_i, X_{i+1}, \dots, X_j$ will be abbreviated as $X_{i:j}$.
The standard inner product $\int f(x) g(x) dx$ will be written as $\langle f, g \rangle$.
For a finite set $X$ with arbitrary elements and real-valued function $h$, the product $\prod_{x \in X} h(x)$ will be written in multi-object exponential form, $h^X$, with $h^\emptyset = 1$ by convention.
The generalized Kronecker delta function over arbitrary arguments is defined as,
\begin{equation}
	\delta_Y(X) \triangleq
	\begin{cases}
		1,		& \text{if } X = Y\\
		0,		& \text{otherwise}
	\end{cases}.
\end{equation}
Finally, we denote the set inclusion function as,
\begin{equation}
	1_Y(X) \triangleq
	\begin{cases}
		1,		& \text{if } X \subseteq Y\\
		0,		& \text{otherwise}
	\end{cases}.
\end{equation}
As a shorthand notation, we adopt the notation $1_Y(x)$ in place of $1_Y(\{x\})$ when $X = \{x\}$.

\subsection{Labeled Multi-object State}\label{sec::background:lmos}
Let $x_k \in \mathbb{X}$ be a random state vector and $l_k \in \mathbb{L}_k$ be a unique label at time step $k$.
Let $\textbf{x}_k = (x_k, l_k) \in \mathbb{X} \times \mathbb{L}_k$ be defined as a labeled object state at time step $k$.
The label space for all objects up to time $k$ is the disjoint union $\mathbb{L}_k = \biguplus^k_{t=0} \mathbb{B}_t$, where $\mathbb{B}_t$ denotes the label space for objects born at time step $t$~\cite{Vo2013}.
The collection of object states $\textbf{x}_{k,1}, \dots, \textbf{x}_{k, N_k}$, known as the multi-object state, is modeled as a labeled~\ac{RFS},
\begin{equation}\label{eq::x_rfs}
	\textbf{X}_k = \{\textbf{x}_{k,1}, \dots, \textbf{x}_{k, n_k}\} \in \mathcal{F}(\mathbb{X} \times \mathbb{L}_k).
\end{equation}
\noindent where $\mathcal{F}(\mathbb{X} \times \mathbb{L}_k)$ is the collection of all finite subsets on $\mathbb{X} \times \mathbb{L}_k$.
Define $\mathcal{L}(\textbf{X}) = \{l : (x, l) \in \textbf{X}\}$ as the set of all labels in $\textbf{X}$.
Since all labels must be unique, we have $\delta_{|\textbf{X}|}(|\mathcal{L}(\textbf{X})|) = 1$,
where $| \cdot |$ denotes the set cardinality.
Define the distinct label indicator as~\cite{Vo2013},
\begin{equation}
	\Delta(\textbf{X}) \triangleq \delta_{|\textbf{X}|}(|\mathcal{L}(\textbf{X})|).
\end{equation}
For the remainder of this paper, we will drop the subscript notation for current time step $k$ and use subscript '$+$' to indicate the next time step ($k+1$).

\subsection{Multi-object Dynamic Model}\label{sec::background::mod}
Between each time step, every object $(x, l) \in \textbf{X}$ can survive with probability $p_s(x,l)$ or can die with probability $q_s(x, l) = 1 - p_s(x,l)$.
If it survives, it evolves to the new state $(x_+, l_+)$ according to the Markov transition density $f_+(x_+|x, l)\delta_l[l_+]$.
The set of surviving objects, $\textbf{W}$, is modeled as a~\ac{LMB}~\ac{RFS} with parameter set $\{(p_s(\textbf{x}), f(\cdot|\textbf{x})): \textbf{x} \in \textbf{X}\}$ distributed according to~\cite{Vo2013, Vo2019},
\begin{equation}
	\textbf{f}_{S,+}(\textbf{W}|\textbf{X}) = \Delta(\textbf{W})\Delta(\textbf{X})1_{\mathcal{L}(\textbf{X})}(\mathcal{L}(\textbf{W}))\left[\Phi(\textbf{W}; \cdot)\right]^\textbf{X},
\end{equation}
where
\begin{multline}
	\Phi(\textbf{W}; x, l) = \sum\limits_{(x+,l+)\in \textbf{W}} \delta_l(l_+)p_s(x,l)f(x_+|x,l) +\\
		[1 - 1_{\mathcal{L}(\textbf{W})}(l)]q_s(x,l).
\end{multline}

The set of newborn objects, $\textbf{B}_+$, is modeled as an~\ac{LMB}~\ac{RFS} with density,
\begin{equation}\label{eq::birth_lmb_def}
	\textbf{f}_{B,+}(\textbf{B}_+) = \Delta(\textbf{B}_+) [1_{\mathbb{B}_+} r_{B, +}]^{\mathcal{L}(\textbf{B}_+)}[1 - r_{B, +}]^{\mathbb{B}_+ - \mathcal{L}(\textbf{B}_+)}p^{\textbf{B}_+}_{B,+}
\end{equation}
where $r_{B,+}(l_+)$ is the probability that an object is born with label $l_+$ and $p_{B,+}$ is the spatial distribution of its kinematic state~\cite{Vo2013, Vo2019}.

The predicted multi-object state, $\textbf{X}_+$ is the superposition of surviving and new born objects, $\textbf{X}_+ = \textbf{W} \cup \textbf{B}_+$.
Under the standard assumption that, conditioned on $\textbf{X}$, objects move, appear, and die independently of one another, the expression for the multi-object transition density is modeled as~\cite{Vo2013, Vo2019},
\begin{equation}
	\textbf{f}_+(\textbf{X}_+|\textbf{X}) = \textbf{f}_{S,+}(\textbf{X}_+ \cap (\mathbb{X} \times \mathbb{L}) | \textbf{X})\textbf{f}_{B,+}(\textbf{X}_+ - (\mathbb{X} \times \mathbb{L})).
\end{equation}

\subsection{Multi-object Observation Model}\label{sec::background::moo}

Suppose the multi-object state $\textbf{X}$ is partially observed by $V \geq 1$ sensors, denoted $s \in \{1, \dots, V\}$, in observation space $\mathbb{Z}^{(s)}$.
Object state $\textbf{x} \in \textbf{X}$ is either detected by sensor $s$ with probability $p_D^{(s)}(\textbf{x})$ or does not generate a measurement with probability $1 - p_D^{(s)}(\textbf{x})$.
If an object is detected, it generates a noisy measurement $z^{(s)}$, modeled by the measurement likelihood function $g^{(s)}(z^{(s)}|\textbf{x})$.
The set of detected points at each sensor are modeled as forming a multi-Bernoulli RFS with parameter set $\{(p_D^{(s)}(\textbf{x}), g^{(s)}(z^{(s)}|\textbf{x})): \textbf{x} \in \textbf{X}\}$, assuming conditional independence on $\textbf{X}$ between each Bernoulli RFS.
In addition, sensor $s$ observes a set of clutter-generated measurements modeled as being sampled according to a Poisson distribution with intensity $\kappa^{(s)}$.
The multi-object observation, $Z^{(s)}$, is modeled as the superposition of object detections and clutter-generated measurements.
It follows that the measurement likelihood is the convolution of the detected multi-Bernoulli RFS distribution and clutter-generated Poisson RFS distributions~\cite{Vo2019}.

Let $m^{(s)}$ be the number of measurements in multi-object observation set $Z^{(s)}$.
We define $\mathbb{J}^{(s)} = \{1, \dots, m^{(s)}\} \subset \mathbb{N}$ as an enumeration index space into $Z^{(s)}$ such that elements $j^{(s)} \in \mathbb{J}^{(s)}$ uniquely index each measurement in $Z^{(s)}$.
The standard single-sensor, multi-object observation likelihood function is given by~\cite{Vo2019},
\begin{equation}
	g^{(s)}(Z^{(s)}|\textbf{X}) \propto \sum\limits_{\theta^{(s)} \in \Theta^{(s)}}1_{\Theta^{(s)}(\mathcal{L}(\textbf{X}))}(\theta^{(s)})  [\psi^{s, \theta^{(s)} \circ \mathcal{L}(\cdot)}_{Z^{(s)}}(\cdot)]^\textbf{X},
\end{equation}
where $\theta^{(s)} \circ \mathcal{L}(\textbf{x}) = \theta^{(s)}(\mathcal{L}(\textbf{x}))$,
$\theta^{(s)}$ is a positive 1:1 function from the object label to the measurement index, and $\theta^{(s)} : \mathbb{L} \rightarrow \{0\} \cup \mathbb{J}^{(s)}$ with $0$ denoting an undetected object label by convention.
For brevity, let the space $\mathbb{J}^{(s)}_0 = \{0\} \cup \mathbb{J}^{(s)}$ be defined as the measurement enumeration space augmented with the element  0.
Note that $\theta^{(s)}$ is only injective for $\theta^{(s)}(l) \in \mathbb{J}^{(s)}$ (i.e., $\theta^{(s)}(i) = \theta^{(s)}(i') > 0 \implies i = i'$) since several object labels may map to element $0$.
The collection $\Theta^{(s)}$ is the set of all $\theta^{(s)}$ maps, with $\Theta^{(s)}(I)$ as the subset of $\Theta^{(s)}$ with domain $I$.
The pseudolikelihood function is given by~\cite{Vo2019},
\begin{equation}\label{eq::meas_likelihood}
	\psi^{s, j^{(s)}}_{Z^{(s)}}(\textbf{x}) =
	\begin{cases}
		\frac{p_D^{(s)}(\textbf{x})g^{(s)}(z^{(s)}_{j^{(s)}} | \textbf{x})}{\kappa^{(s)}(z^{(s)}_{j^{(s)}})} 	& j^{(s)} \in \mathbb{J}^{(s)}\\
		1 - p_D^{(s)}(\textbf{x}) 																			& j^{(s)} = 0
	\end{cases}.
\end{equation}

We present an augmented version of the multi-sensor abbreviated notation provided in~\cite{Vo2019},
\begin{align*}
	\mathbb{J}_0 &\triangleq \mathbb{J}_0^{(1)} \times \dots \times \mathbb{J}_0^{(V)},
	&&J \triangleq (j^{(1)}, \dots, j^{(V)}),\\
	\Theta 	&\triangleq \Theta^{(1)} \times \dots \times \Theta^{(V)},\hfill
	&&\theta 	\:\triangleq (\theta^{(1)}, \dots, \theta^{(V)}),\\
	Z 			&\triangleq (Z^{(1)}, \dots, Z^{(V)}),\hfill
	&& Z_J \triangleq (z^{(1)}_{j^{(1)}}, \dots, z^{(V)}_{j^{(V)}}),\\
	1_{\Theta(I)}(\theta) &\triangleq \prod\limits^V_{s=1} 1_{\Theta^{(s)}(I)}(\theta^{(s)}),\hfill
	&&\psi^J_Z(\textbf{x}) \triangleq \prod\limits^V_{s=1} \psi^{s, j^{(s)}}_{Z^{(s)}}(\textbf{x})
\end{align*}
where $\mathbb{J}_0$ is defined as the \textit{augmented multi-sensor, multi-observation index space}, with index $\{0\}$ added as notation to denote missed detections.
Let $J \in \mathbb{J}_0$ be defined as a \textit{multi-sensor index tuple} that contains a $0-augmented$ measurement index for every sensor.
By conditional independence of the sensors, the multi-sensor multi-object likelihood is of the same form as the single-sensor multi-object likelihood~\cite{Vo2019},
\begin{equation}
	\begin{split}
		g(Z|\textbf{X}) &= 	\prod\limits^V_{s = 1} g^{(s)}(Z^{(s)} | \textbf{X})\\
						&\propto \sum\limits_{\theta \in \Theta} \delta_{\Theta(\mathcal{L}(\textbf{X}))}(\theta)[\psi^{\theta \circ \mathcal{L}(\cdot)}_{Z}(\cdot)]^\textbf{X}.
	\end{split}
\end{equation}

\subsection{Multi-sensor $\delta$-GLMB Recursion}\label{sec::background::msglmb}

A $\delta$-\ac{GLMB} density is a labeled multi-object density of the form~\cite{Vo2013},
\begin{equation}
	\boldsymbol{\pi}(\textbf{X}) = \Delta(\textbf{X}) \sum\limits_{I, \xi} w^{(I, \xi)} \delta_I[\mathcal{L}(\textbf{X})][p^{(\xi)}(\cdot)]^\textbf{X},
\end{equation}
where $I \in \mathcal{F}(\mathbb{L})$ is a finite subset of object labels, $\xi \in \Xi$ represents a history of multi-sensor association maps, and $p^{(\xi)}(\cdot, l)$ is a probability density on $\mathbb{X}$.
Each hypothesis weight $w^{(I, \xi)}$ is non-negative and sum to $1$.

Under the multi-object dynamic and measurement model described in Sections~\ref{sec::background::mod} and~\ref{sec::background::moo}, the $\delta$-\ac{GLMB} density is a conjugate prior with itself.
This results in a closed-form solution to the multi-object Bayes filtering recursion.
Since the multi-sensor likelihood has the same form as the single-sensor likelihood function, it follows that the multi-sensor multi-object posterior is also a $\delta$-\ac{GLMB} and is given by~\cite{Vo2013},
\begin{equation}\label{eq::glmb_posterior}
	\begin{split}
	&\boldsymbol{\pi}_+(\textbf{X}) \propto\\
	&\Delta(\textbf{X})\sum\limits_{I, \xi, I_+, \theta_+} w^{(I, \xi)} w_{Z_+}^{(I, \xi, I_+, \theta_+)}\delta_{I_+}[\mathcal{L}(\textbf{X})]\left[p^{(\xi, \theta_+)}_{Z_+}\right]^\textbf{X}.
	\end{split}
\end{equation}
\noindent where
\begin{subequations}
	\begin{align}
		\begin{split}
			&w_{Z_+}^{(I, \xi, I_+, \theta_+)} = 1_{\Theta_+(I_+)}(\theta_+)[1 - \bar{P}^{(\xi)}_s]^{I-I_+}[ \bar{P}^{(\xi)}_s]^{I \cap I_+}\\
			&\qquad\qquad\quad\times\left[1 - r_{B,+}\right]^{\mathbb{B}_+ - I_+} r_{B, +}^{\mathbb{B}_+\cap I_+}[\bar{\psi}^{(\xi, \theta_+)}_{Z_+}]^{I_+}
		\end{split}\label{eq::glmb_update::w}\\
		&\bar{P}^{(\xi)}_s(l) = \langle p^{(\xi)}(\cdot, l), p_s(\cdot, l) \rangle\\
		&\bar{\psi}^{(\xi, \theta_+)}_{Z_+}(l_+) = \langle \bar{p}^{(\xi)}_s(\cdot, l_+), \psi^{(\theta_+(l_+))}_{Z_+}(l_+) \rangle\label{eq::glmb_update::psi_bar}\\
		\begin{split}
			&\bar{p}^{(\xi)}_+(\cdot, l_+) = 1_{\mathbb{B}_+}(l_+)p_{B,+}(x_+, l_+)\\
			&\qquad\qquad + 1_{\mathbb{L}}(l_+) \frac{\langle p_s(\cdot, l_+)f_+(x_+|\cdot, l_+), p^{(\xi)}(\cdot, l_+) \rangle}{\bar{P}^{(\xi)}_s(l_+)}
		\end{split}\label{eq::glmb_update::p+}\\
		&p^{(\xi, \theta_+)}_{Z_+}(x_+, l_+) = \frac{\bar{p}^{(\xi)}_+(x_+, l_+)\psi^{(\theta_+(l_+))}_{Z_+}(x_+, l_+)}{\bar{\psi}^{(\xi, \theta_+)}_{Z_+}(l_+)}\label{eq::glmb_update::pZ+}
	\end{align}\label{eq::glmb_update}
\end{subequations}
Initial implementation details for the $\delta$-\ac{GLMB} filter were introduced in~\cite{Vo2013} for Gaussian mixture and sequential Monte Carlo (i.e., particle) spatial distributions.
For tractability,~\cite{Vo2013} truncates~\ac{GLMB} multi-target exponentials using the ranked assignment and K-shortest path algorithms in the prediction and update steps respectively.
A more efficient implementation of the $\delta$-\ac{GLMB} filter is provided in~\cite{Vo2017} which combines the filtering recursion into a joint predict-update step and presents a stochastic truncation algorithm based on Gibbs sampling.
Multi-sensor implementation details were initially discussed in~\cite{Papi2016} using an iterated-corrector approach, and later formalized in~\cite{Vo2019} for the joint multi-sensor formulation.
%
%
\subsection{$\delta$-GLMB Measurement Association Probabilities}
The probability that a single-sensor's measurement, $z^{(s)}_{j_+^{(s)}, +}$, is associated with existing targets in the $\delta$-\ac{GLMB} posterior density is proportional to the sum of the posterior hypotheses weights where measurement $z^{(s)}_{j_+^{(s)}, +}$ is associated with a target.
More formally,
\begin{equation}
	r_{A,+}(j_+^{(s)}) \propto \sum\limits_{I, \xi, I_+, \theta_+} 1_{\theta_+^{(s)}}(j_+^{(s)}) w^{(I, \xi)} w_{Z_+}^{(I, \xi, I_+, \theta_+)},
\end{equation}
where $1_{\theta_+^{(s)}}(j_+^{(s)})$ is the inclusion function denoting that measurement $z^{(s)}_{j_+^{(s)}, +}$ is associated with a target in sensor $s$'s mapping $\theta_+^{(s)}$ \cite{Lin2016}.
By notation, let $r_{A,+}(0) = 0$ which is intuitive as it suggests that a missed detection did not associate with any tracks in the existing hypotheses.
The probability that the multi-sensor measurement tuple $J_+$ is unassociated with existing targets in the $\delta$-\ac{GLMB} posterior density is approximated by,
\begin{equation}\label{eq::unassoc_prob}
    r_{U,+}(J_+) \propto \left[1 - r_{A,+}\right]^{J_+}.
\end{equation}
%
%
\section{Birth Label Truncation Criterion}\label{sec::truncation}

Although the size of the birth distribution can be large, in most applications many newborn components do not contribute significantly to the posterior $\delta$-\ac{GLMB}.
This concept is formalized in the following theorem.
\begin{theorem}\label{theorem::l1_dist}
    Let $\textbf{f}_{B,+}$ be the birth \ac{LMB} \ac{RFS} density generated from all possible multi-sensor measurement tuples resulting in a birth label space $\mathbb{B}_+$.
    Let $\textbf{f}'_{B,+}$ be a truncated birth \ac{LMB} \ac{RFS} density of $\textbf{f}_{B,+}$ with a birth label space $\mathbb{B}'_+ \subseteq \mathbb{B}_+$ formed such that,
    \begin{equation}\label{eq::theorem_l1}
        \mathbb{B}'_+ = \{ l_+ \in \mathbb{B}_+ : r_{B, +}(l_+) \geq \epsilon\},
    \end{equation}
    for truncation threshold $\epsilon \geq 0$.
    Let $\mathbb{H} \in \mathcal{F}(\mathbb{L}) \times \mathcal{F}(\mathbb{L} \cup \mathbb{B}_+) \times (\Xi \times \Theta_+)$ and $\mathbb{H}' \in \mathcal{F}(\mathbb{L}) \times \mathcal{F}(\mathbb{L} \cup \mathbb{B}'_+) \times (\Xi \times \Theta_+)$ be the hypotheses in the posterior $\delta$-\ac{GLMB} using the original and truncated birth sets respectively.
    The L1-truncation error between $\mathbb{H}$ and $\mathbb{H}'$ is upper bounded by,
    \begin{equation}
        |\pi_{\mathbb{H}} - \pi_{\mathbb{H}'}| \leq \sum\limits_{(I, \xi, I_+, \theta_+)\in \mathbb{H} - \mathbb{H}'} 1_{\Theta_+(I_+)}(\theta_+) K^{|I_+|} \epsilon^{N_{\mathbb{T}_+}(I_+)},
    \end{equation}
    where $N_{\mathbb{T}_+}(I_+) = |I_+ \cap (\mathbb{B}_+ \setminus \mathbb{B}'_+)|$ is the number of truncated newborn labels in $I_+$, and
	$K$ is a positive upper bound, $0 \leq \bar{\psi}^{(\xi, \theta_+)}_{Z_+}(l_+) \leq K, \; \forall I_+ \in \mathcal{F}(\mathbb{L} \cup (\mathbb{B}_+ \setminus \mathbb{B}'_+))$ with $l_+ \in I_+$ and $\theta_+ \in \Theta_+(I_+)$.
\end{theorem}
The proof of Theorem~\ref{theorem::l1_dist} is provided in Appendix~\ref{sec::appendix::proof_l1_dist}.
Intuitively, the truncation criterion is interpreted as pruning components that are expected to have a low existence probability in the $\delta$-\ac{GLMB} posterior density.
The weight of any hypothesis containing a newborn birth label $l_+$ will be multiplicatively proportional to $r_{B,+}(l_+)$ (by inspection and simplification of Equation~(\ref{eq::glmb_update::w})).
If this value is low, then the weight of every hypothesis containing $l_+$ will also be low, resulting in a low existence probability for label $l_+$ \cite{Vo2014}.
The significance of Theorem~\ref{theorem::l1_dist} is that the L1-distance between the truncated and untruncated $\delta$-\ac{GLMB} posterior distributions is upper bounded by a positive polynomial in $\epsilon$.
As $\epsilon \rightarrow 0$, the L1-distance $|\pi_{\mathbb{H}} - \pi_{\mathbb{H}'}| \rightarrow 0$.

%% file: sections/problem_formulation.tex
\section{Multi-sensor Adaptive Birth Density}\label{sec::problem}

In many practical applications, the birth prior, $p_{B,+}$, is uninformative in one or more states of the state space.
If the birth prior is informative for an application, then a static birth technique could be used to construct a birth~\ac{LMB}~\ac{RFS} density according to Equation~(\ref{eq::birth_lmb_def}).
Instead in an adaptive birth procedure, we aim to construct a birth~\ac{LMB}~\ac{RFS} density for the next time step using multi-sensor measurements from the current time step.
Let each unique birth label be defined as $l_+ = (k+1, J)$, resulting in a birth label space $\mathbb{B}_+ = \{ (k+1, J) : \forall J \in \mathbb{J}_0 \}$.
Since each label is generated using a unique measurement tuple $J$, by construction there is a known bijective mapping between birth labels and measurement tuples, $\theta': \mathbb{B}_+ \rightarrow \mathbb{J}_0$.
Because of this, we will interchange $\theta'(l_+)$ and $J$ as necessary for clarity or brevity.

Extending Reuter, et al.'s suggestion in~\cite{Reuter2014}, let the birth~\ac{LMB}~\ac{RFS} density be constructed with parameter set,
\begin{equation}
	\textbf{f}_{B,+} = \left\{ \left(r_{B,+}(l_+), p_{B,+}(\cdot, l_+ | Z_J) \right) \right\}_{l_+ \in \mathbb{B}_+},
\end{equation}
where $r_{B,+}(l_+)$ and $p_{B,+}(\cdot, l_+ | Z_J)$ are the birth probability and spatial distribution of label $l_+$ respectively.
Since the measurements $Z_J$ were observed at time step $k$ and the birth~\ac{LMB} density is incorporated at time step $k+1$, the spatial distribution of each Bernoulli component is modeled using the posterior density of the birth prior using measurement $Z_J$, then predicted to time step $k+1$.
This procedure reduces the risk of divergence between measurements at the next time step and the components in the birth~\ac{LMB}.
For example, if a set of sensors observe a highly maneuverable aircraft at time step $k$, and the discrete sampling interval between time steps $k$ and $k+1$ is large, then the target may have moved substantially by the time it is next observed.
If this dynamic uncertainty is not accounted for in the birth component's spatial distribution, then the measurements at the next time step may not associate with the birth component and could lead to track switching or track divergence.
The dynamic model uncertainty is incorporated into the spatial distribution using the Chapman-Kolmogorov equation \cite{Mahler2007},
\begin{equation}\label{eq::pred_birth_post}
    p_{B,+}(x_+, l_+ | Z_J) = \int f_+(x_+|x, l) p_B(x, l_+ | Z_J) dx,
\end{equation}
where, by Bayes rule,
\begin{align}
    p_B(x, l_+ | Z_J) &= \frac{p_B(x, l_+) \psi^{J}_Z(x, l_+)}{\bar{\psi}^{J}_{Z}(l_+)} \label{eq::spatial_distr}\\ 
    \bar{\psi}^{J}_{Z}(l_+) &= \langle p_B(\cdot, l_+), \psi^{J}_{Z}(\cdot, l_+)\rangle. \label{eq::psi_bar}
\end{align}

There are two desirable properties of a newborn label $l_+$ generated from a multi-sensor adaptive birth procedure;
(1) measurements used to generate the newborn component should be minimally associated with existing tracks in the current $\delta$-\ac{GLMB} posterior density (i.e., high $r_U(J)$) and,
(2) multi-sensor measurements used to generate newborn components should be consistent across sensors for all $x$ (i.e., high $\bar{\psi}^{J}_{Z}(l_+)$).
Using this insight, we model the birth probability for each Bernoulli component as,
\begin{equation}\label{eq::birth_prob}
	r_{B,+}(l_+) = \min\left(r_{B, \text{max}}, \hat{r}_{B,+}(l_+)\lambda_{B,+} \right),
\end{equation}
where $r_{B, \text{max}} \in [0, 1]$ is the maximum existence probability of a newborn target,
$\lambda_{B,+}$ is the expected number of target births at time step $k+1$,
and the effective birth probability is given as,
\begin{equation}\label{eq::birth_prob_hat}
    \hat{r}_{B,+}(l_+) = \frac{r_U(J) \bar{\psi}^{J}_{Z}(l_+)}{\sum\limits_{J' \in \mathbb{J}_0} r_U(J')\bar{\psi}^{J'}_{Z}(l_+)}.
\end{equation}

\subsection{Multi-sensor Birth Set Size Complexity Analysis}\label{sec::pf::complexity}
The labels in $\mathbb{B}_+$ consist of \textit{all combinations} of multi-sensor measurements.
For $m^{(s)}$ measurements at each sensor, the number of birth labels in $\mathbb{B}_+$ is equal to $\prod^V_{s=1} (m^{(s)} + 1)$.
Assuming the maximum number of measurements from all of the sensors is $m$, the worst-case complexity is $O(m^V)$.
This quickly becomes intractable as the number of sensors or the number of measurements per observation set increases.
For example, in a scenario with 5 sensors and 15 measurements each, there would be $16^{5} \simeq$ 1 million newborn labels.

%% file: sections/gibbs_sampler.tex
\section{Gibbs Sampling Truncation of the Multi-sensor Adaptive Birth Density}\label{sec::gibbs}

The truncation criterion described in Theorem~\ref{theorem::l1_dist} requires evaluation of Equation~(\ref{eq::theorem_l1}) over an exponentially large number of newborn components (see Section~\ref{sec::pf::complexity}).
Instead, we aim to derive a technique to efficiently sample labels according to a categorical distribution proportional to,
\begin{equation}\label{eq::orig_sampling_distr}
	p(l_+) \propto r_{B,+}(l_+).
\end{equation}
This is intuitive as it states that newborn labels should be sampled proportionally to their birth probability.
From Equation~(\ref{eq::birth_prob}), $r_{B,+}(l_+)$ is given as $\hat{r}_{B,+}(l_+)\lambda_{B,+}$ upper bounded by $r_{B,max}$.
Assuming that the value of $r_{B, max}$ is chosen large enough that it does not inhibit the differentiation between labels with a large and small value of $\hat{r}_{B,+}(l_+)\lambda_{B,+}$, the sampling distribution can be simplified to,
\begin{equation}\label{eq::sampling_distr}
	p(l_+) \propto \hat{r}_{B,+}(l_+) \propto r_U(J)\bar{\psi}^{J}_{Z}(l_+).
\end{equation}
Directly sampling from the categorical distribution $p(l_+)$, still requires evaluation over an exponential number of possible birth labels.
By construction of the birth set, there exists a 1:1 mapping between $l_+$ and $J$.
This means that sampling from $p(l_+)$ is equivalent to sampling multi-sensor measurement tuples from the joint distribution, $p(J) = p(j^{(1)}, \dots, j^{(V)})$.
Intuitively, this can be seen as sampling multi-sensor measurement tuples that will be used to generate labels in the adaptive birth \ac{LMB}~\ac{RFS} density.
From this insight, we construct a Gibbs sampler to efficiently generate elements from the multi-sensor measurement adaptive birth set according to the truncation criterion established in Theorem~\ref{theorem::l1_dist}.
\begin{theorem}\label{prop::cnd_likelihood}
	The joint likelihood $p(l_+)$ is proportional to the conditional likelihood,
	\begin{equation}
		p(l_+) \propto p(j^{(s)} | J^{-s}), \forall s = 1, \dots, V
	\end{equation}
	where $J^{-s} = (j^{(1)}, \dots, j^{(s-1)}, j^{(s+1)}, \dots, j^{(V)})$.
	For a given sensor $s$ the conditional likelihood is proportional to,
	\begin{equation}\label{eq::cdn_likelihood}
		p(j^{(s)} | J^{-s}) \propto \left(1 - r_A(j^{(s)})\right) \bar{\psi}^{J}_{Z}(l_+),
	\end{equation}
\end{theorem}
The proof of Theorem~\ref{prop::cnd_likelihood} is found by combining all of the terms in the product $\left[1 - r_A\right]^{J}$ that are not a function $j^{(s)}$ into the normalizing constant.

\begin{algorithm}[t!]
\caption{Multi-sensor Adaptive Birth Gibbs Sampler}\label{alg::gibbs}
	\begin{algorithmic}[1]
		\renewcommand{\algorithmicrequire}{\textbf{Input:}}
		\renewcommand{\algorithmicensure}{\textbf{Output:}}
		\Require
			\Statex $Z$, $r_U$, $p_{B,+}$, $r_{B,\max}$, $\lambda_{B,+}$, $T$
		\Ensure
			\Statex $\textbf{f}'_{B,+}$
		\LeftComment Gibbs Sampling Truncation
		\State $\mathbb{B}'_+ = \emptyset$
		\State $J = \left(0, \dots, 0\right)$
		\For{$t = 1, \dots, T$}
			\For {$s \in \text{shuffle}(\{1, \dots, V\})$}
				\For{$j^{(s)} \in \mathbb{J}^{(s)}_0$}
					\State evaluate $p(j^{(s)} | J^{-s})$ using Theorem~(\ref{prop::cnd_likelihood})
				\EndFor
				\State $j^{(s')} \sim p(j^{(s')} | J^{-s'})$ 
				\State $J = (j^{(1)}, \dots, j^{(s-1)}, j^{(s')}, j^{(s+1)}, \dots, j^{(V)})$
		\EndFor
			\State $\mathbb{B}'_+ = \mathbb{B}'_+ \cup \theta'^{-1}(J)$
		\EndFor
		\LeftComment Constructing the Birth LMB
		\State $\textbf{f}'_{B,+} = \emptyset$
		\For {$l_+ \in \mathbb{B}'_+$}
			\State $J = \theta'(l_+)$
			\State construct $p_{B,+}(\cdot, l_+ | Z_J)$ according to Equation~(\ref{eq::pred_birth_post})
			\State evaluate $r_{B,+}(l_+)$ according to Equation~(\ref{eq::birth_prob})
			\State $\textbf{f}'_{B,+} = \textbf{f}'_{B,+} \cup \{ ( r_{B,+}(l_+), p_{B,+}(\cdot, l_+ | Z_J) ) \}$
		\EndFor
	\end{algorithmic}
\end{algorithm}

Implementation of the Gibbs sampling procedure is provided in Algorithm~\ref{alg::gibbs} where $T$ is the number of Gibbs iterations.
Starting from any initial state $J$, the proposed Gibbs sampler, defined by the conditional likelihood in Theorem~\ref{prop::cnd_likelihood}, converges to the target distribution (Equation~(\ref{eq::orig_sampling_distr})).
For a finite state discrete Gibbs sampler, irreducibility with respect to the target distribution is a sufficient condition for convergence \cite{Roberts1994}.
On the $t$'th Gibbs iteration, the transition probability from any state $J^{t-1}$ to state $J^{t}$ is the product of the conditional likelihoods.
If the detection probability of every sensor $s$ is $p_D^{(s)}(\textbf{x}) < 1$, then the transition probability from any state to the all-missed detection state is strictly positive.
Similarly for any state $J$, if $r_A(j^{(s)}) < 1$ for all $z^{(s)}_{j^{(s)}}$ and $\bar{\psi}^J_{Z}(l_+)$ is continuous, then the transition probability from the all-missed detection state to any state is strictly positive.
Using a similar approach as the proof of \cite[Proposition 4]{Vo2014}, the two-step transition probability from any state to any other state is strictly positive since it can always transition to the all-missed detection state first.

The Gibbs sampler in Algorithm~\ref{alg::gibbs} does not require a burn-in period.
Every unique solution can be directly used since we do not need to discard samples until we reach a stationary distribution.
It is important to ensure that the Gibbs sampler is encouraged to seek diverse solutions and prevent it from stalling at local maxima, sometimes referred to as high probability islands.
In Algorithm~\ref{alg::gibbs}, we randomize the order of the sensor indexes on each Gibbs iteration to encourage exploration.
Other approaches include annealing or tempering techniques that modify the stationary distribution~\cite{Geyer1995, Neal2001} or through restarting logic if a unique solution has not been found after a certain number of samples.

In this paper, we initialize $J$ using the all-missed detection tuple, $J=(0, \dots, 0)$.
When no prior information is available, starting with the all-missed detection tuple allows for more diverse exploration of unique solutions and is less likely to initialize to a high probability island.

\begin{remark}
	In Theorem~\ref{prop::cnd_likelihood}, if the association probability is high for any measurement index, $r_A(j^{(s)}) \approx 1$, then the sampling probability will be close to zero, $p(j^{(s)} | J^{-s}) \approx 0$ regardless of the value of $\bar{\psi}^J_{Z}(l_+)$.
		Since in practice $r_A$ needs to be calculated at the current time step for all $j^{(s)}$ before the adaptive birth procedure, a pre-pruning technique can be employed before Algorithm~\ref{alg::gibbs} by sampling over the subset of measurement indexes $\{ j^{(s)} \in \mathbb{J}^{(s)}_0 : r_A(j^{(s)}) > \tau \}$ for all $s \in V$, where $\tau \in [0, 1]$ is a user-specified maximum association probability threshold.
		This removes measurement indexes from being considered in the sampler that are unlikely since $p(j^{(s)} | J^{-s}) \approx 0$.
\end{remark}

The implementation of Algorithm~\ref{alg::gibbs} consists of computing or constructing three probability distributions which, in practice, can be challenging to evaluate; (1) the sampling distribution $p(j^{(s)} | J^{-s})$, (2) the newborn component's spatial distribution $p_{B,+}(\cdot, l_+ | Z_J)$ and (3) the newborn component's existence probability $r_{B,+}(l_+)$.
One of the primary challenges associated with evaluating these distributions is computing the inner product $\bar{\psi}^J_Z$.
Section~\ref{sec::mc} addresses these challenges using a Monte Carlo approximation and Section~\ref{sec::gm} addresses these challenges under a linear Gaussian model assumption.

%% file: sections/monte_carlo.tex
\section{Monte Carlo Approximation}\label{sec::mc}
\subsection{Monte Carlo Approximation of $\bar{\psi}^J_Z(l_+)$}\label{sec::mc::psi_bar}
By Equation~(\ref{eq::psi_bar}), since the birth prior density $p_B$ is a valid probability density on $x$, then the inner product is the expected value of $\psi^J_Z(\cdot, l_+)$ such that, $\bar{\psi}^J_Z(l_+) = \mathbb{E}_{p_B}[\psi^J_Z(\cdot, l_+)]$.
Using Monte Carlo integration, the expected value of any arbitrary function can be approximated using a set of independently and identically distributed (i.i.d.) samples $\{x_n\}^{N_p}_{n=1}$ from its associated probability density.
That is \cite[Chapter 11]{Bishop2006},
\begin{equation}\label{eq::orig_mc_exp}
	\bar{\psi}^J_Z(l_+) \approx \frac{1}{N_p} \sum\limits_{n=1}^{N_p} w_n(x_n, l_+)
\end{equation}
where $x_n \sim p_B$ and $w_n = \psi^J_Z(x_n, l_+)$.
If an application has an informative prior density, $p_B$, then the expected value can simply be approximated by sampling from $x_n \sim p_B$ and evaluating each sample under the pseudolikelihood.
However as discussed in Section~\ref{sec::problem}, the birth prior distribution $p_B$ is often uninformative.
Sampling $x_n$ from an uninformative distribution would result in a large estimator variance for a fixed number of samples.

Instead, we can use an importance sampling procedure to sample from a more informative proposal distribution, $x_n \sim q$ \cite[Chapter 11]{Bishop2006}.
This leads to the updated importance weights being proportional to the likelihood ratio $p_B/q$,
\begin{equation}\label{eq::wn_orig}
	w_n = \psi^J_Z(x_n, l_+)\frac{p_B(x_n, l_+)}{q(x_n, l_+)}.
\end{equation}
One choice of a proposal distribution is similar to the procedure provided in \cite{Ristic2012}.
First assume independence between the observable states ($x_o$) and unobservable states ($x_u$) such that the birth prior decomposes to $p_B(x, l_+) = p^o_B(x_o, l_+)p^u_B(x_u, l_+)$, where $p^o_B$ and $p^u_B$ are prior densities on the observable and unobservable state spaces respectively.
Our objective is to construct a more informative proposal distribution for the observable ($q^o$) and unobservable ($q^u$) state spaces such that, $q(x, l_+) = q^o(x_o, l_+)q^u(x_u, l_+)$.
Since $x_u$ is unobserved, $q^u$ is modeled as $q^u(x_u, l_+) = p^u_B(x_u, l_+)$.
Since $x_o$ is observed, $q^o$ can be modeled such that a measurement $z^{(s)}_{j^{(s)}}$ can be considered a random sample from the density $g^{(s)}(\cdot | x_o)$.
We choose $j^{(s')}$ to be any non-missed detection measurement in the measurement tuple $J$ (i.e., $j^{(s')} > 0$ for $j^{(s')} \in J$).
Let the observation function for sensor $s'$ be given by $z^{(s')}_{j^{(s')}} = h^{(s')}(x_o) + \eta$,
where $h^{(s')}$ is an invertible function with a differentiable inverse, $h^{(s'),-1}$, and $\eta$ is additive, zero-mean white Gaussian noise such that $\eta \sim \mathcal{N}(0, R^{(s')})$.
Under this model the proposal distribution for observable states is,
\begin{equation}
	q^o(x_o, l_+) = \mathcal{N}(x_o | h^{(s'),-1}(z^{(s')}_{j^{(s')}}), \tilde{H}^{(s')}R^{(s')}\tilde{H}^{(s'),T}),
\end{equation}
where $\tilde{H}^{(s')}$ is the Jacobian of $h^{(s'),-1}$.
This results in importance weights,
\begin{equation}\label{eq::imp_weights}
	w_n = \frac{\psi^{J}_Z(x_n, l_+)p^o_B(x_{n,o}, l_+)}{\mathcal{N}(x_{n,o} | h^{(s'), -1}(z^{(s')}_{j^{(s')}}), \tilde{H}^{(s')}R^{(s')}\tilde{H}^{(s'),T})}.
\end{equation}
Substituting the importance weights from Equation~(\ref{eq::imp_weights}) into Equation~(\ref{eq::orig_mc_exp}) results in an approximation of $\bar{\psi}^J_Z(l_+)$ under this choice of proposal distribution.
%
\subsection{Monte Carlo Approximation of $p(j^{(s)} | J^{-s})$}\label{sec::mc::sample}
\begin{algorithm}[t]
	\caption{Monte Carlo Approximation of $p(j^{(s)}|J^{-s})$}\label{alg::monte_carlo}
	\begin{algorithmic}[1]
	\renewcommand{\algorithmicrequire}{\textbf{Input:}}
	\renewcommand{\algorithmicensure}{\textbf{Output:}}
	\Require
		\Statex $Z$, $J$, $r_A$, $s$, $N_p$
	\Ensure
		\Statex $p(j^{(s)}|J^{-s})$
	\State  $j^{(s')} \sim Cat\left(\left\{j^{(s')} \in J : j^{(s')} > 0\right\}\right)$
	\For{$n = 1, \dots, N_p$}
		\State $x_{n} \sim q(x, l_+)$
		\State compute $w_n$ using Equation~(\ref{eq::imp_weights})
	\EndFor
	\State evaluate $p(j^{(s)}|J^{-s})$ using Equations~(\ref{eq::orig_mc_exp}) and (\ref{eq::cdn_likelihood}) 
	\end{algorithmic}
\end{algorithm}
Substituting the proposed approximation of $\bar{\psi}^J_Z(l_+)$ presented in Section~\ref{sec::mc::psi_bar} into the sampling distribution in Equation~(\ref{eq::cdn_likelihood}) directly results in the Monte Carlo approximation of $p(j^{(s)} | J^{-s})$.
The implementation of the Monte Carlo approximation of $p(j^{(s)}|J^{-s})$ is given in Algorithm~\ref{alg::monte_carlo} where $j^{(s')} \sim Cat\left(\left\{j^{(s')} \in J : j^{(s')} > 0\right\}\right)$ denotes sampling a non-missed detection index from the tuple $J$ with equal probability.
%
\subsection{Monte Carlo Approximation of $p_{B,+}(\cdot, l_+ | Z_J)$}\label{sec::mc::spatial}
\begin{algorithm}[t]
	\caption{Monte Carlo Approximation of $p_{B,+}(\cdot, l_+ | Z_J)$}\label{alg::mc_spatial}
	\begin{algorithmic}[1]
	\renewcommand{\algorithmicrequire}{\textbf{Input:}}
	\renewcommand{\algorithmicensure}{\textbf{Output:}}	
	\Require
		\Statex $Z$, $J$, $N_p$
	\Ensure
		\Statex $\{w_{n,+}, x_{n,+}\}^{N_p}_{n=1}$
	\State  $j^{(s')} \sim Cat\left(\left\{j^{(s')} \in J : j^{(s')} > 0\right\}\right)$
	\For{$n = 1, \dots, N_p$}
		\State $x_{n} \sim q(x, l_+)$
		\State compute $w_n$ using Equation~(\ref{eq::imp_weights})
	\EndFor
	\State normalize and resample $\{w_n, x_n\}^{N_p}_{n=1}$
	\For{$n = 1, \dots, N_p$}
		\State $x_{n,+} \sim f_+(x_{n,+} | x_n, l_+)$
		\State $w_{n,+} = 1/N_p$
	\EndFor
\end{algorithmic}
\end{algorithm}
Our objective is to generate a set of samples, $\{w_n, x_n\}^{N_p}_{n=1}$, that are distributed according to Equation~(\ref{eq::spatial_distr}) such that,
\begin{equation}\label{eq::mc_prior_approx}
	p_B \left( x, l_+ | Z_J\right) \propto p_B(x, l_+) \psi^J_Z(x, l_+) \approx \sum^{N_p}_{n = 1} w_n \delta_{x}(x_n).
\end{equation}
This can be accomplished by sampling $x_n \sim p_B$ with weights $w_n = \psi^J_Z(x_n, l_+)$.
If $p_B(x, l_+)$ is informative and easy to sample from, then this technique can be easily implemented.
If it is not easy to sample from, an importance sampling technique similar to that proposed in Section~\ref{sec::mc::psi_bar} can be used.
To help avoid degeneracy, a resampling procedure such as multinomial, stratified, or systematic resampling can be employed \cite{Li2015}.

Given the Monte Carlo approximation of $p_B \left( x, l_+ | Z_J\right)$ in Equation~(\ref{eq::mc_prior_approx}), the posterior density predicted to the birth time from Equation~(\ref{eq::pred_birth_post}) is given by,
\begin{equation}
	p_{B,+} \left(x_+, l_+ | Z_J\right) \approx \sum^{N_p}_{n=1} w_n f_+(x_+ | x_n)\delta_{x}(x_n).
\end{equation}
This results in a new set of Monte Carlo supports $\{w_{n,+}, x_{n,+}\}^{N_p}_{n=1}$, where $\{w_{n,+}\}^{N_p}_{n=1} = \{w_n f_+(x_+ | x_n)\}^{N_p}_{n=1}$ and $\{x_{n,+}\}^{N_p}_{n=1} = \{x_n\}^{N_p}_{n=1}$.
Note that in this approach, only the weights are updated in accordance with the transition density.
Equivalently, if the sample weights $w_n$ are uniform, the set of predicted sample states can instead be more naturally sampled from the transition density, $x_{n,+} \sim f(x_+ | x_n)$, with uniform weights $w_{n,+} = 1/N_P$.
The implementation of the Monte Carlo sampling of $p_{B,+} \left(x_+, l_+ | Z_J\right)$ is given in Algorithm~\ref{alg::mc_spatial}.
%
\subsection{Monte Carlo Approximation of $r_{B,+}(l_+)$}\label{sec::mc::birth_prob}
The Monte Carlo approximation for the birth probability is directly given by first substituting the proposed approximation of $\bar{\psi}^J_Z(l_+)$ presented in Section~\ref{sec::mc::psi_bar} into Equation~(\ref{eq::birth_prob_hat}) and then substituting this result into Equation~(\ref{eq::birth_prob}).
%
\subsection{Complexity Analysis Using Monte Carlo Approximation}
Sampling newborn labels using the Gibbs sampler proposed in Algorithm~\ref{alg::gibbs} with the proposed Monte Carlo approximation of the sampling distribution in Section~\ref{sec::mc::sample} is $O(mTN_pV^2)$, which is linear in configurable parameters $(T, N_p)$, linear in the maximum number of measurements $m$ and quadratic in the number of sensors $V$.
Constructing the newborn label's spatial distribution and birth probabilities with the proposed Monte Carlo approximations from Section~\ref{sec::mc::spatial} and Section~\ref{sec::mc::birth_prob} respectively is $O(|\mathbb{B}'_+| N_p V)$, which is linear in configurable parameters and the number of birth labels in the truncated birth \ac{LMB} density.

%% file: sections/gaussian.tex
\section{Gaussian Models}\label{sec::gm}
In this section, a closed-form solution to the expected value $\bar{\psi}^J_Z(l_+)$, sampling distribution $p(j^{(s)}|J^{-s})$, spatial distribution $p_{B,+} \left( \cdot, l_+ | Z_J\right)$ and birth probability $r_{B,+}(l_+)$ is derived under a linear Gaussian assumption.
For this section, let the single-sensor measurement likelihood and birth prior be modeled or approximated as a Gaussian and detection probability such that,
\begin{align}
	g^{(s)}(z_{j^{(s)}}^{(s)} | x) &= \mathcal{N}(z_{j^{(s)}}^{(s)}; H^{(s)}x, R^{(s)})\\
	p_B(x, l_+) &= \mathcal{N}(x;\mu_0, P_0)\\
	p_D^{(s)}(x, l_+) &= p_D^{(s)}.
\end{align}
Additionally let the state transition be modeled or approximated as a linear Gaussian system such that,
\begin{equation}
	f_+(x_+ | x) = \mathcal{N}(x_+; Fx, Q).
\end{equation}
%
%
\subsection{Gaussian Evaluation of $\bar{\psi}^J_Z(l_+)$}
\begin{theorem}\label{theorem::gm_psibar}
	Under the Gaussian model from Section~\ref{sec::gm}, the expected value of the pseudolikelihood is,
	\begin{align}
		\begin{split}\label{eq::gauss_psibar}
			\bar{\psi}^J_Z(l_+) &= 	\left[\det(P_0)\det(M_J)\prod\limits_{s'=1}^V\left((2\pi)^{n^{(s')}_z}\det(R^{(s')})\right)\right]^{-\frac{1}{2}}\\
									&\times \left[\prod\limits^V_{\substack{s'=1\\j^{(s') = 0}}} (1-p^{(s')}_D)\right]
									\left[\prod\limits^V_{\substack{s'=1\\j^{(s') > 0}}} \frac{p^{(s')}_D}{\kappa^{(s)}(z^{(s')}_{j^{(s')}})}\right]\Phi_J,
		\end{split}
	\end{align}
	where,
	\begin{align}
		\Phi_J &= \exp \left\{-\frac{1}{2}(c_J - b_J^TM_J^{-1}b_J)\right\}\label{eq::gauss_phi}\\
		M_J &= P_0^{-1} + \sum\limits^V_{\substack{s'=1\\j^{(s')} > 0}}H^{(s'),T}R^{(s'),-1}H^{(s')}\label{eq::gauss_M}\\
		b_J &= P_0^{-1}\mu_0 + \sum\limits^V_{\substack{s'=1\\j^{(s')} > 0}}H^{(s'),T}R^{(s'),-1}z_{j^{(s')}}^{(s')}\label{eq::gauss_bJ}\\
		c_J &= \mu_0^TP^{-1}_0\mu_0 + \sum\limits^V_{\substack{s'=1\\j^{(s')} > 0}}z_{j^{(s')}}^{(s'),T}R^{(s'),-1}z_{j^{(s')}}^{(s')}\label{eq::gauss_cJ}.
	\end{align}
\end{theorem}
The proof of Theorem~\ref{theorem::gm_psibar} is provided in Appendix~\ref{sec::appendix::proof_gm_theorem0}.
%
\subsection{Gaussian Evaluation of $p(j^{(s)}|J^{-s})$}
\begin{theorem}\label{theorem::gm_sample_distr}
	Under the Gaussian model from Section~\ref{sec::gm}, the sampling distribution is,
	\begin{subequations}\label{eq::gm_sampling}
		\begin{align}
			&p(j^{(s)} = 0|J^{-s}) \propto \left(1 - p_D^{(s)}\right)\det(M_{J^{-s}})^{-\frac{1}{2}}\Phi_{J^{-s}}\label{eq::gm_sampling_eq0}\\
			&\begin{aligned}\label{eq::gm_sampling_gt0}
				p(j^{(s)} > 0|J^{-s}) \propto
					&\left[(2\pi)^{n^{(s)}_z}\det(M_J)\det(R^{(s)})\right]^{-\frac{1}{2}}\\
					&\times \left(1 - r_A(j^{(s)})\right)
					\left[\frac{p^{(s)}_D}{\kappa^{(s)}(z^{(s)}_{j^{(s)}})}\right]
					\Phi_J
			\end{aligned}
		\end{align}
	\end{subequations}
	Where $M_{J^*}$ and $\Phi_{J^*}$ are given by Equation~(\ref{eq::gauss_M}) and Equation~(\ref{eq::gauss_phi}) respectively over the domain $J^*$.
\end{theorem}
The proof of Theorem~\ref{theorem::gm_sample_distr} can be found in Appendix~\ref{sec::appendix::proof_gm_theorem1}.
%
\subsection{Gaussian Evaluation of $p_{B,+}(\cdot, l_+ | Z_J)$}
\begin{theorem}\label{theorem::gm_spatial_distr}
	Under the Gaussian model from Section~\ref{sec::gm}, the spatial distribution is,
	\begin{equation}\label{eq::gm_spatial}
		p_{B,+} \left( x_+, l_+ | Z_J\right) = \mathcal{N}(x_+; \mu', P'),
	\end{equation}
	where,
	\begin{align*}
		&\mu' = FM_J^{-1}b_J,
		&&P'= FM_J^{-1}F^T + Q.
	\end{align*}
	Where $M_{J}$ and $b_J$ are given by Equation~(\ref{eq::gauss_M}) and Equation~(\ref{eq::gauss_bJ}) respectively.
\end{theorem}
The proof of Theorem~\ref{theorem::gm_spatial_distr} can be found in Appendix~\ref{sec::appendix::proof_gm_spatial_distr}.
The result of Theorem~\ref{theorem::gm_spatial_distr} is intuitive because of the conjugate prior nature of the Gaussian distribution.
The expression $\mu'$ is the predicted multi-sensor generalized least squares solution with a prior $p_B$ and $P'$ is the estimate covariance of $\mu'$~\cite[Chapter~2.1 \& 2.4]{Crassidis2011}.
Hence, newborn components in the birth set are the generalized least squares solution of the non-missed detection elements in the measurement tuple predicted to the next time step.
%
\subsection{Gaussian Evaluation of $r_{B,+}(l_+)$}
The Gaussian approximation for the birth probability is directly given by substituting the expression for $\bar{\psi}^J_Z(l_+)$ from Equation~(\ref{eq::gauss_psibar}) into Equation~(\ref{eq::birth_prob_hat}) and replacing this value for $\hat{r}_{B,+}$ into Equation~(\ref{eq::birth_prob}).
%
\subsection{Complexity Analysis Using Gaussian Models}
Worst case evaluation of the sampling distribution given by Equation~(\ref{eq::gm_sampling}) requires $V + 3$ matrix inversions, $6V + 5$ matrix multiplications and $3V + 4$ matrix additions assuming the availability of reuse for terms between Equations~(\ref{eq::gm_sampling_eq0}) and Equations~(\ref{eq::gm_sampling_gt0}).
Using this sampling distribution in Algorithm~\ref{alg::gibbs} results in $O(mTV^2)$ matrix operations to sample labels for the truncated birth \ac{LMB} density.
To reduce the runtime complexity, if the birth prior parameters $\mu_0, P_0$ and the measurement covariances $R^{(s)}$ are not time varying, the inverses $P_0^{-1}$ and $(R^{(1),-1}, \dots, R^{(V),-1})$ can be precomputed and stored at initialization.
This removes all but 2 of matrix inversions that need to be computed at each Gibbs sampler iteration ($M^{-1}_J$ and $M^{-1}_{J^{-s}}$).
Similarly, the terms $P^{-1}_0\mu_0$, $H^{(s'),T}R^{(s'),-1}$ and $H^{(s'),T}R^{(s'),-1}H^{(s')}$ can be precomputed for all sensors.

Construction of the spatial distribution given by Equation~(\ref{eq::gm_spatial}) requires $V + 1$ matrix inversion, $4V+5$ matrix multiplications and $2V + 3$ matrix additions.
In addition to aforementioned precomputations, to reduce complexity, the values of $M_J^{-1}$ and $M_J^{-1}b_J$ computed during the evaluation of $p(j^{(s)}|J^{-s})$ in Equation~(\ref{eq::gm_sampling}), can be reused by Equation~(\ref{eq::gm_spatial}) resulting in only 3 matrix multiplications and 1 matrix addition to construct the spatial distribution.
Construction of the birth probability requires $V + 1$ matrix inversions, $6V + 5$ matrix multiplications and $6V + 3$ matrix additions.
To reduce complexity, the term $b^T_JM^{-1}_Jb_J$ can be stored when evaluating the spatial distribution in Equation~(\ref{eq::gm_sampling}) resulting in Equation~(\ref{eq::gm_spatial}) only needing $2V$ matrix multiplications and 1 matrix addition to construct the spatial distribution.
Using this in Algorithm~\ref{alg::gibbs} results in $O(|\mathbb{B}'_+| V)$ matrix operations to construct the truncated birth \ac{LMB} density.

%% file: sections/simulations.tex
\section{Simulation Examples}\label{sec::sim}

In this section, we demonstrate the truncation accuracy and scalability of the proposed multi-sensor adaptive birth procedure through two simulated scenarios.
In both scenarios, the target birth locations are unknown \textit{a priori} rendering it impossible to accurately describe a static birth prior.
Scenario 1 shows the performance of the proposed Monte Carlo approximation described in Section \ref{sec::mc} and Scenario 2 highlights the performance of the Gaussian solution described in Section \ref{sec::gm}.
The multi-sensor adaptive birth models were incorporated into the multi-sensor iterated corrector formulation of the \ac{LMB} and $\delta$-\ac{GLMB} filters and the results are compared against the same filters using a uniform birth procedure.

Both scenarios contained a time-varying number of targets.
The target state comprised of planar 2D position and velocity, $x = [p_{x}, \dot{p}_x, p_y, \dot{p}_y]^T$.
The targets followed a constant velocity transition model, $f_+(x_+|x) = \mathcal{N}(x_+; Fx, Gw)$ where the transition matrix and process noise matrix were given by,
\begin{equation*}
    F =
    \begin{bmatrix}
            1 & \Delta\\
            0 &     1
    \end{bmatrix},\qquad
    G = 
    \begin{bmatrix}
        \frac{\Delta^2}{2}\\
        \Delta
    \end{bmatrix},
\end{equation*}
respectively and $\Delta$ is the discrete-time sampling interval \cite{Li2003}.
The discrete-time x and y acceleration white noises were, $w = [5, 5]^T\;m/s^2$.
Both simulations were run for 100 seconds with $\Delta = 1$ second.
The survival probability for each target was set to $p_s(x,l) = 0.99$.
In both scenarios, the detection probability for all sensors was modeled as constant with $p_D^{(s)}(x,l) = 0.95$.
Clutter was modeled as Poisson distributed with intensity $\kappa^{(s)}(\mathbb{Z}^{(s)}) = \lambda_c^{(s)} \mathcal{U}(\mathbb{Z}^{(s)})$
where $\mathcal{U}(\mathbb{Z}^{(s)})$ is the uniform distribution over $\mathbb{Z}^{(s)}$, and
$\lambda^{(s)}_c = 15$.

The uniform birth model used for comparison generated 100 labeled birth components at each time step, with means uniformally distributed throughout 2D state space.
The birth probability of all components was $r_{B,+}(l_+) = 0.1$.
The spatial distribution of each birth component was modeled as a Gaussian, $p_{B,+}(x_+, l_+) = \mathcal{N}(x_+; m_{B,+}, P_{B,+})$, with
\begin{align}
    \begin{split}
        m_{B,+} \in &\{-2000, -1000, \dots, 11000, 12000\}\\
                \times &\{0, 0, \dots, 0, 0\}\\
                \times &\{-2000, -1000, \dots, 11000, 12000\}\\
                \times &\{0, 0, \dots, 0, 0\},\\
    \end{split}
\end{align}
and $P_{B,+} = \text{diag}(\sigma_{p}^2, \sigma_{\dot{p}}^2, \sigma_{p}^2, \sigma_{\dot{p}}^2)$.
For Scenario 1 and Scenario 2, $\sigma_{p}$ was 300 m and 250 m and  $\sigma_{\dot{p}}$ was 20 m/s and 50 m/s respectively.
For the Monte Carlo filters in Scenario 1, 1000 samples were drawn from each Gaussian component.

The adaptive birth Gibbs sampler used 100 Gibbs Samples for Scenario 1 and 1000 Gibbs Samples for Scenario 2 with a simple restarting procedure that reset the current solution to the all-missed detection measurement tuple every 5 and 100 iterations respectively to encourage exploration.
Scenario 1 used significantly fewer Gibbs samples to reduce simulation time and did not result in a significant impact on tracker performance.
A maximum birth probability $r_{B, max} = 1.0$, expected birth rate of $\lambda_{B,+} = 0.5$ and maximum association probability $\tau = 0.01$ was used for both scenarios.
To prevent components being born using only a single sensor measurement and missed detections from all remaining sensors, e.g., $J = (1, 0, 0, 0)$, only sampled measurement tuples that had at least 2 non-missed detection measurement indexes were used to construct the newborn birth set.
For example, if the measurement tuples $J = (1, 1, 0, 0)$, $J = (1, 0, 0, 0)$, $J = (0, 1, 0, 0)$ were sampled, then only $J = (1, 1, 0, 0)$ would be used to birth a newborn component.
This significantly reduced the number of newborn components that were provided to the tracker that had similar spatial distributions.
\begin{figure}[t]
    \subfloat[Scenario 1]{\includegraphics[width=0.48\textwidth, height=6cm]{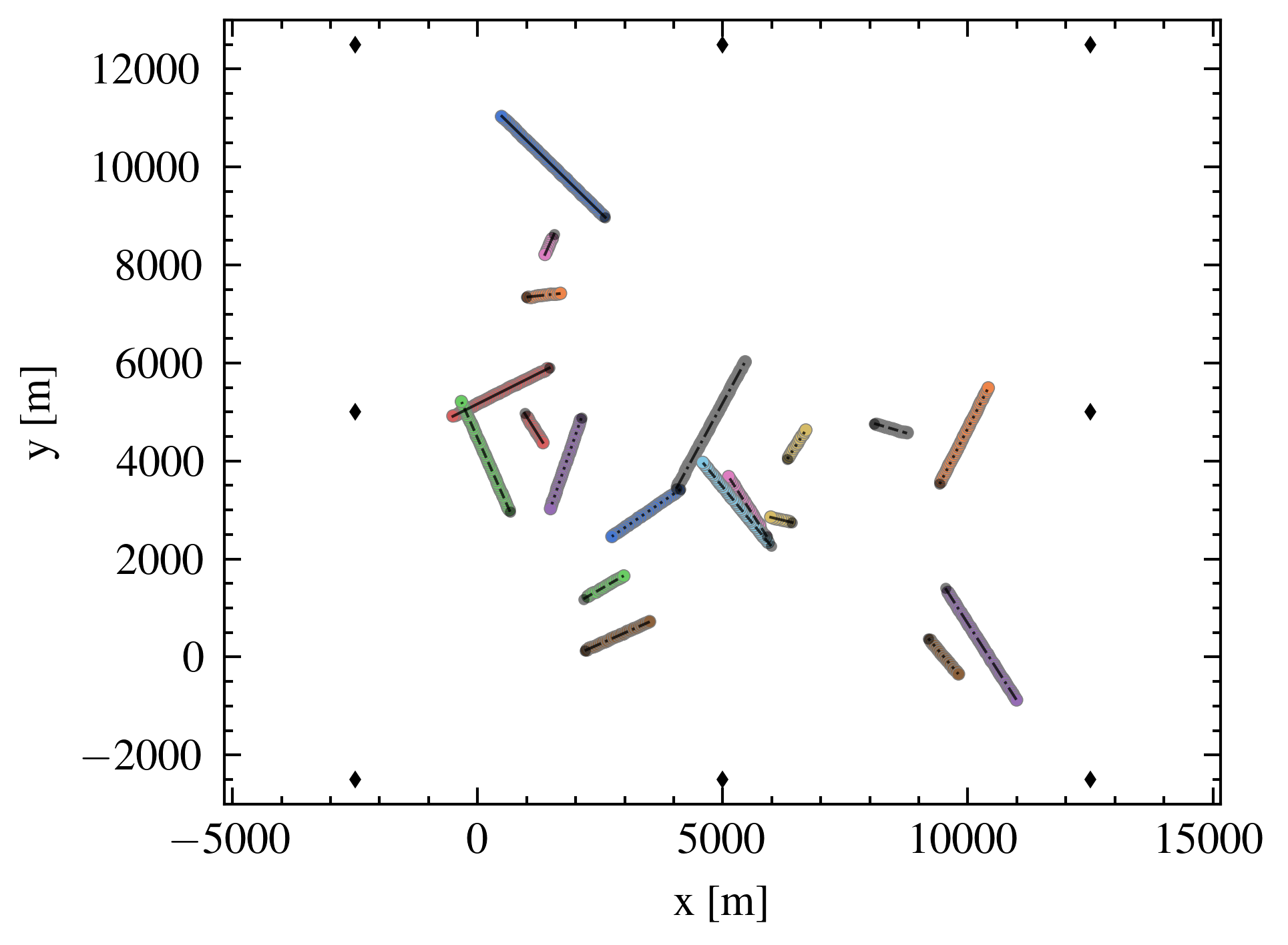}\label{fig::qual::smc}}
    \newline
    \subfloat[Scenario 2]{\includegraphics[width=0.48\textwidth, height=6cm]{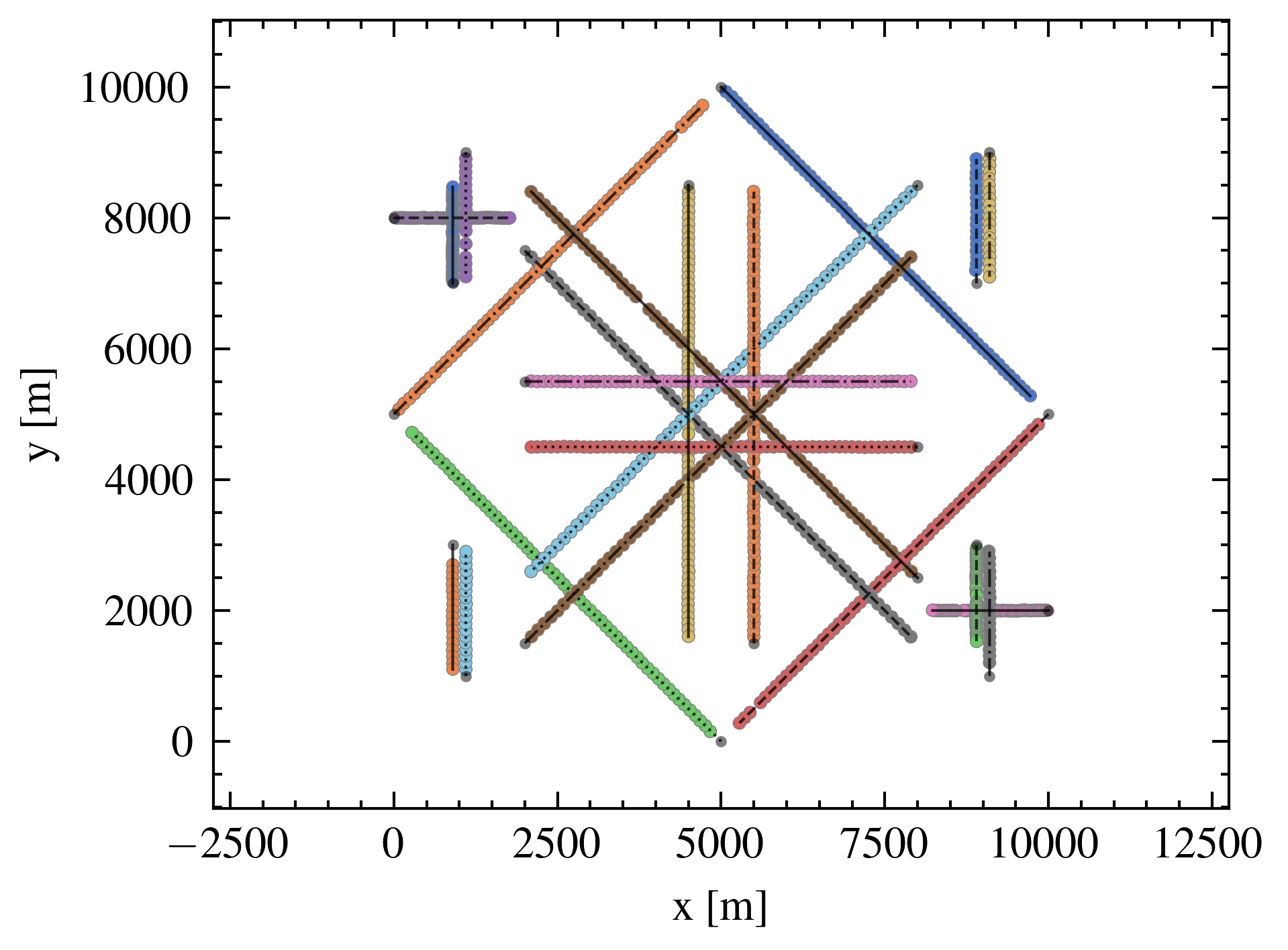}\label{fig::qual::gm}}
    \caption{Single observation of target trajectories (black lines), their birth locations (black circles) and labeled state estimates of a \ac{LMB} (Scenario 1) and $\delta$-\ac{GLMB} (Scenario 2) filter using the proposed adaptive birth model (colored circles where each color indicates a unique label). Diamonds in Scenario 1 represent the locations of the bearing-range sensors.}
    \label{fig::qual}
\end{figure}
%
\subsection{Scenario 1: Bearing-range Tracking Example}\label{sec::sim::bearing}
The first scenario contained 8 bearing-range sensors tracking a randomly generated number of targets.
Every 5 seconds, a number between 0 and 3 was sampled uniformly to determine how many true targets were born.
The positions of these targets were uniformly sampled in position state space over the domain $[0, 10000]\;m$.
The speed of each target was fixed at $50\;m/s$, and the velocity heading was uniformly sampled in the domain $[-\pi, \pi]\;rad$.
The randomness of the trajectories in this scenario render it difficult to describe a static birth prior.
Figure \ref{fig::qual::smc} shows a single observation of the randomly generated trajectories and state estimate results from the \ac{LMB} filter using the proposed Monte Carlo multi-sensor birth model.

If an object was detected, the bearing-range measurement
$z^{(s)} = [\alpha^{(s)}, r^{(s)}]^T$
was observed according to the single-target measurement likelihood
$g(z^{(s)}|x) = \mathcal{N}(z^{(s)}; h^{(s)}(z^{(s)}, x^{(s)}), R^{(s)})$
where
$x^{(s)} = [p^{(s)}_x, p^{(s)}_y]^T$
is the position of the sensor, $R^{(s)} = \text{diag}(0.25^2, 10^2)$ and,
\begin{align}
    h_\alpha^{(s)}(x, x^{(s)}) &= \arctan\left(\frac{p^{(s)}_x - p_x}{p^{(s)}_y - p_y}\right)\\
    h_r^{(s)}(x, x^{(s)}) &= \sqrt{(p^{(s)}_x - p_x)^2 + (p^{(s)}_y - p_y)^2}.
\end{align}

An uninformative uniform prior distribution $p^o_{B}(x_o, l_+) = \mathcal{U}(\mathbb{X}_o)$ was used where $\mathbb{X}_o$ is the observable portion of state space $\mathbb{X}$.
The unobservable velocities were sampled from zero-mean Gaussian distribution $p^u_{B}(x_u, l_+) = \mathcal{N}(x_u; 0, \begin{bsmallmatrix}\sigma_{\dot{p}}^2 & 0\\ 0 & \sigma_{\dot{p}}^2\end{bsmallmatrix})$ with $\sigma_{\dot{p}} = 20\;m/s^2$.

\subsection{Scenario 2: Linear Position Tracking Example}\label{sec::sim::position}

The second scenario contained 8
linear XY-position sensors tracking a maximum of 22 simultaneous targets.
The birth locations of each target were fixed, but are sparse and no two targets were born from the same location.
Figure~\ref{fig::qual::gm} shows the target trajectories and a single observation of the state estimator results from the $\delta$-\ac{GLMB} filter using the proposed Gaussian multi-sensor birth model.

If an object was detected, the XY-position measurement $z^{(s)} = [p_x, p_y]^T$ was observed according to the single-target measurement likelihood $g(z^{(s)}|x) = \mathcal{N}(z^{(s)}; H^{(s)}x, R^{(s)})$ with, $R^{(s)} = \text{diag}(10^2, 10^2)$ and,
$H^{(s)} = \begin{bsmallmatrix}1 & 0\\0 & 1\end{bsmallmatrix} \otimes \begin{bsmallmatrix}1 & 0\\0 & 0\end{bsmallmatrix}$ where $\otimes$ denotes the Kronecker product.

The birth prior was modeled as a Gaussian with an uninformative covariance in position such that, $p_B(x, l_+) = \mathcal{N}(x; \mu_0, P_0)$ with $\mu_0 = \left[0, 0, 0, 0\right]^T$ and $P_0 = \text{diag}(100000^2, 50^2, 100000^2, 50^2)$.

\subsection{Results}\label{sec::sim::results}

\begin{figure}[t!]
    \subfloat[Scenario 1]{\includegraphics[width=0.48\textwidth, height=6cm]{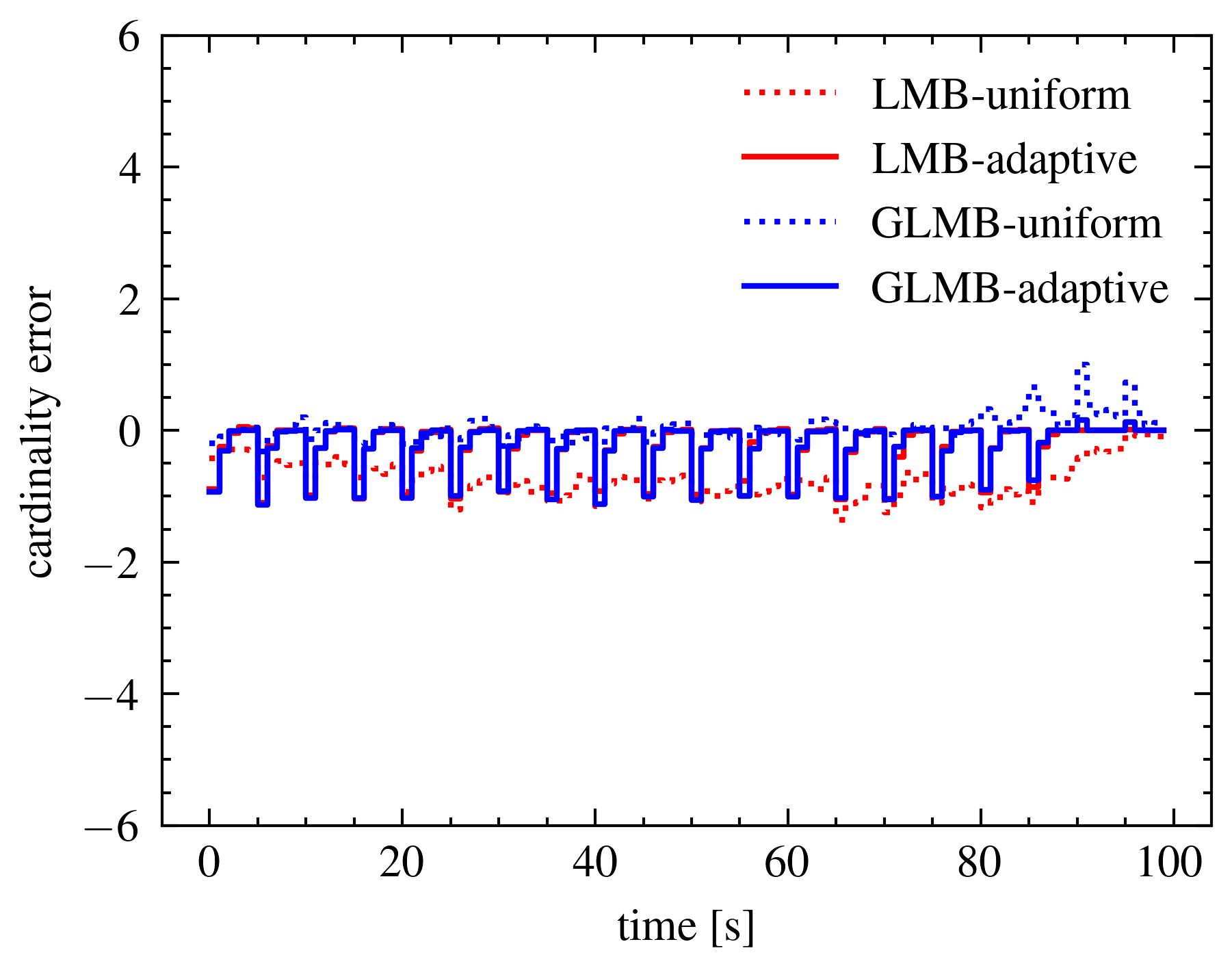}\label{fig::smc::results::card}}
    \newline
    \subfloat[Scenario 2]{\includegraphics[width=0.48\textwidth, height=6cm]{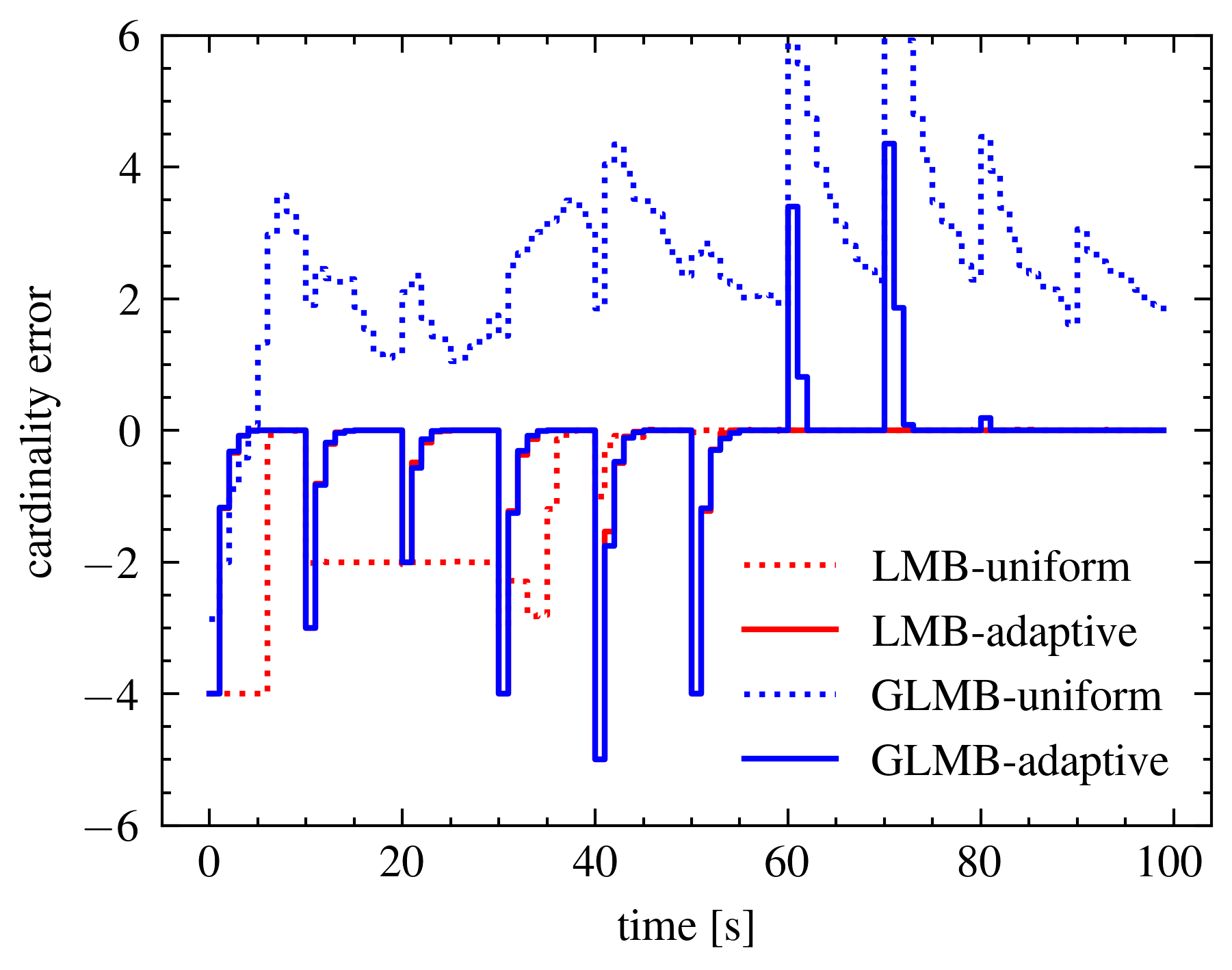}\label{fig::gm::results::card}}
    \caption{Scenario 1 (top) and Scenario 2 (bottom) cardinality errors. Solid and dotted lines represent the average values over 100 Monte Carlo iterations.}
    \label{fig::results::card}
\end{figure}

\begin{figure}[t!]
    \subfloat[Scenario 1]{\includegraphics[width=0.48\textwidth, height=6cm]{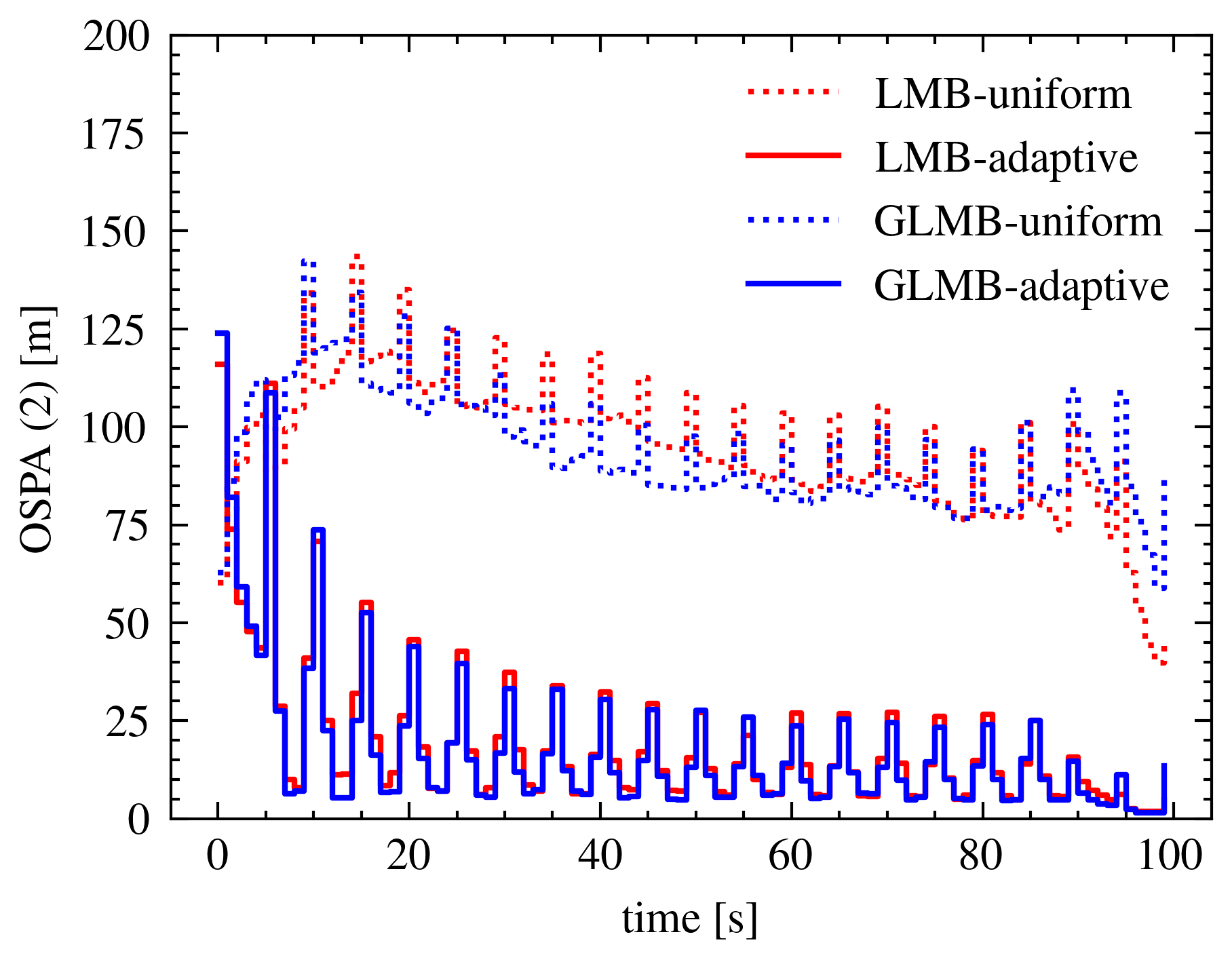}\label{fig::smc::results::ospa}}
    \newline
    \subfloat[Scenario 2]{\includegraphics[width=0.48\textwidth, height=6cm]{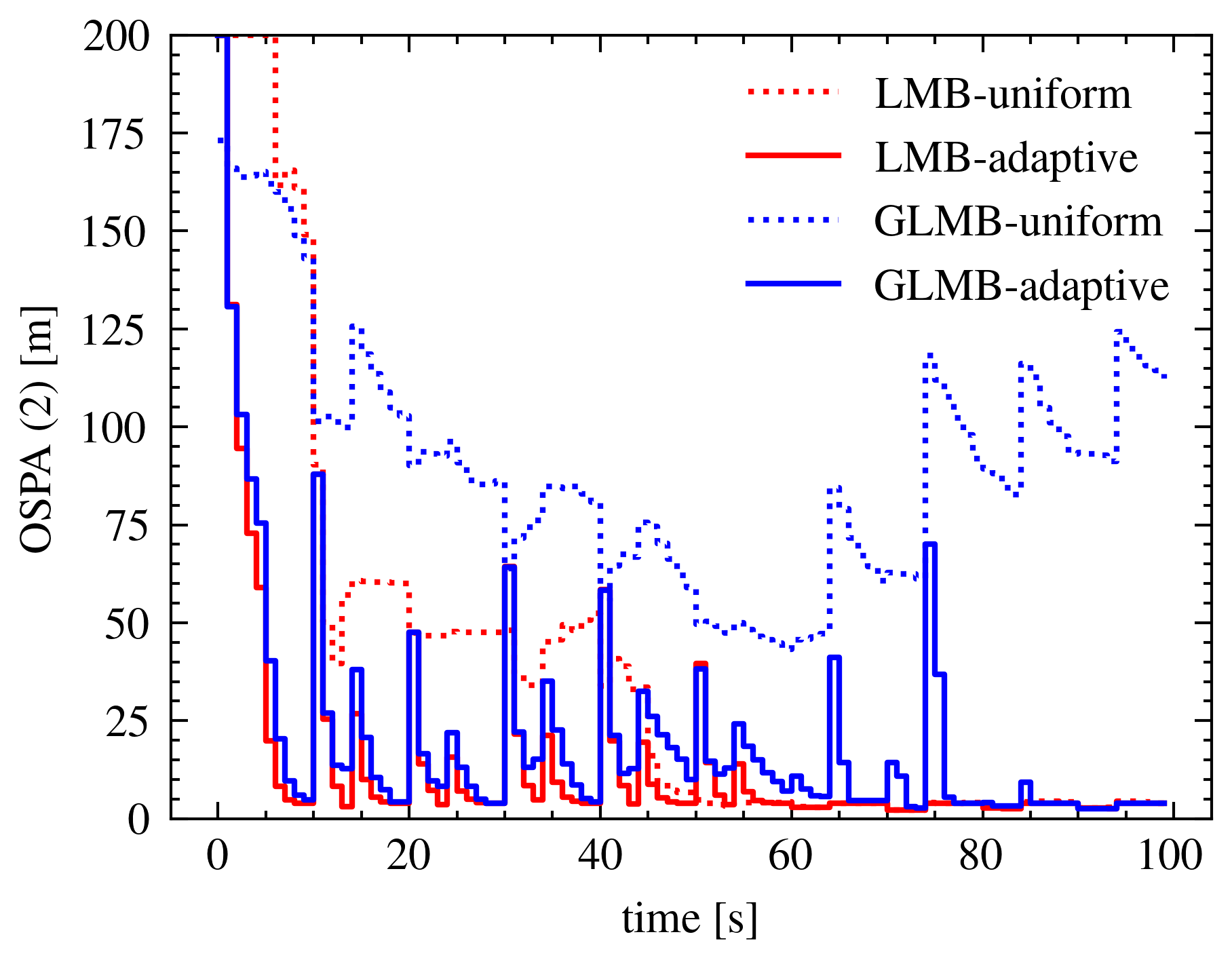}\label{fig::gm::results::ospa}}
    \caption{Scenario 1 (top) and Scenario 2 (bottom) OSPA(2) results. Solid and dotted lines represent the average value over 100 Monte Carlo iterations.}
    \label{fig::results::ospa}
\end{figure}

The cardinality and state estimation accuracy was quantified using the \ac{OSPA}(2) metric \cite{Beard2020}.
The \ac{OSPA}(2) metric was computed using a distance cutoff value of $200\;m$, a distance order of $1.0$, a sliding window length of $5$ and an expanding window weight power of $0$.
Additionally, the cardinality error was computed as the difference between the estimated cardinality and the true cardinality at each time step such that a positive or negative value indicates an overestimate or underestimate of the true cardinality respectively.
The results were averaged over 100 Monte Carlo iterations.

As seen in Figure~\ref{fig::results::card}, the proposed multi-sensor adaptive birth algorithms consistently tracked the correct cardinality in both scenarios.
The short-duration cardinality errors are due to the time-delayed nature of the proposed adaptive birthing procedure which results in at least a 1-time step lag before a target can be born.
These cardinality errors typically do not persist for more than 1 time step.
The two overbiased cardinality spikes in Scenario 2's results using the $\delta$-\ac{GLMB} filter with the adaptive birth model correspond to when several targets simultaneously die in the environment.
This indicates that $\delta$-\ac{GLMB} state extraction method persisted its estimate of these targets for an additional time step, where the \ac{LMB} did not.

In contrast, the uniform birth models resulted in a larger cardinality error.
In general, when using the uniform birth model, the \ac{LMB} filter underbiased the cardinality as opposed to the $\delta$-\ac{GLMB} filter which overbiased the cardinality.
This is likely because of the following.
At each time step, the uniform birth model added 100 newborn components with low existence probabilities.
Because many of these components are far from the true targets, these components will likely be missed detected by one or more sensors.
Since the detection probability is high, this results in very low probability of existence for these labels.
After an update, our implementation of the \ac{LMB} filter prunes components with an existence probability less than $1e-3$ and caps the total number of components to $100$, dropping the lowest weighted components.
Similarly, our implementation of the $\delta$-\ac{GLMB} filter prunes hypotheses with a probability less than $1e-5$ and caps the total number of hypotheses to $1000$.
This results in the \ac{LMB} filter being quicker to remove low weighted birth components where as the $\delta$-\ac{GLMB} persists these as ghost tracks for longer, especially if they correlate with some clutter measurements.

Figure~\ref{fig::results::ospa} shows that the \ac{LMB} and $\delta$-\ac{GLMB} filter using the proposed multi-sensor adaptive birth models outperformed the filters using uniform birth models in terms of the \ac{OSPA}(2) metric.
The filters using the uniform adaptive birth model resulted in a large \ac{OSPA}(2) error, which is mostly driven by the cardinality error discussed above.
In contrast, using the proposed adaptive birth model led to a significantly reduced \ac{OSPA}(2) error.

Table~\ref{tab::num_births} shows a comparison between the maximum possible number of multi-sensor measurement tuples and the size of the newborn birth set after running the proposed Gibbs truncation algorithm.
The results in this table were averaged over time and over Monte Carlo iterations.
With 8 sensors and a sensor clutter rate of 15 clutter returns on average, the number of possible multi-sensor measurement tuples is many orders of magnitude more than what a typical tracking system could handle.
The proposed Gibbs sampler successfully truncated this birth set by several orders of magnitude without sacrificing filter performance.
As opposed to the uniform birth model which generated $100$ newborn components on every time step, the proposed adaptive birth procedure generated approximately an average of $1.40$ and $0.42$ newborn components every time step for Scenario 1 and 2 respectively.

\begin{table}[]
    \caption{Gibbs Truncation Efficiency}\label{tab::num_births}
    \begin{tabular}{|c|c|c|c|}
    \hline
    \textbf{Scenario}  & \textbf{Avg Max Births}  & \textbf{Filter Type} & \textbf{Avg Num Birth Comps} \\ \hline
    \multirow{2}{*}{1} & \multirow{2}{*}{6.15e10}  & LMB                  & 1.37                        \\ \cline{3-4} 
                       &                          & GLMB                 & 1.42                        \\ \hline
    \multirow{2}{*}{2} & \multirow{2}{*}{5.21e11} & LMB                  & 0.42                        \\ \cline{3-4} 
                       &                          & GLMB                 & 0.42                        \\ \hline
    \end{tabular}
\end{table}

%% file: sections/conclusions.tex
\section{Conclusion}\label{sec::conclusions}

This paper provided a formal definition for the multi-sensor multi-object adaptive birth distribution for labeled \ac{RFS} filters.
We then showed that the number of components in this birth distribution is exponential in the number of sensors.
To alleviate this, a truncation criterion is established for a \ac{LMB} birth density.
The proposed truncation criterion is shown to have a bounded L1 error in the \ac{GLMB} posterior density.
Using this truncation criterion, we derived an efficient Gibbs sampler that produces a truncated multi-sensor measurement-generated \ac{LMB} birth density.
We then provided Monte Carlo and Gaussian implementations of our approach and verified our results using two simulated scenarios.
The results of the simulations showed that our proposed approach can accurately birth components in challenging scenarios where static birth models are not feasible.
Additionally, the results showed that our proposed Gibbs sampler can successfully truncate the components to just those that are likely to be from true targets.
Future work is being conducted to expand this approach to unlabeled birth intensities and to investigate the effects of Gibbs tempering and herding \cite{Wolf2020}.

%% file: sections/appendix_technical_lemmas.tex
\section{Supplemental Lemmas}\label{sec::appendix::lemmas}

\begin{lemma}\label{lemma::M_positive_definite}
The matrix $M_J$ from Equation~(\ref{eq::gauss_M}) is symmetric positive definite.
\end{lemma}

\begin{proof}
Since $P_0$ and $R^{(s')}$ are symmetric positive definite by construction, then the inverses $P_0^{-1}$, and $R^{(s'),-1}$ are also symmetric positive definite~\cite[Chapter~7.1]{Horn2012}.
For any non-zero vector $u \in \mathbb{R}^{n^{(s')}_z}$ and by the commutative property of matrices,
\begin{equation*}
u^TH^{(s'),T}R^{(s'),-1}H^{(s')}u = (uH^{(s')})^TR^{(s'),-1}(uH^{(s')}).
\end{equation*}
If $H^{(s')}u$ maps to a non-zero vector, $v \in \mathbb{R}^{n^{(s')}_z}$, then we know $v^TR^{(s'),-1}v>0$, since $R^{(s'),-1}$ is positive definite.
If $H^{(s')}u$ has a null space other than the zeros vector, then $v^TR^{(s'),-1}v=0$.
Therefore, $(uH^{(s')})^TR^{(s'),-1}(uH^{(s')}) \ge 0$ which is symmetric positive semi-definite.
Since the addition of a symmetric positive definite matrix and a symmetric semi-positive definite matrix is positive definite, then $M_J$ must be symmetric positive definite~\cite[Chapter~7.1]{Horn2012}.
\end{proof}

\begin{lemma}
\label{lemma::gm1}
Under the linear Gaussian assumptions in Section~\ref{sec::gm}, the following equivalency holds
\begin{multline*}
	\prod\limits_{\substack{s'=1\\j^{(s')} > 0}}^V g^{(s')}(z_{j^{(s')}}^{(s')} | x)p_{B}(x) =\\
		\left[(2\pi)^{n_x}\det(P_0)\prod\limits_{\substack{s'=1\\j^{(s')} > 0}}^V (2\pi)^{n^{(s')}_z}\det(R^{(s')})\right]^{-\frac{1}{2}}\\\
		\times \exp\left\{-\frac{1}{2} (c_J-b_J^T M_J^{-1} b_J) \right\}\\
		\times \exp\left\{-\frac{1}{2}(x-M_J^{-1}b_J)^TM(x-M_J^{-1}b_J)\right\}
\end{multline*}
where $\det(\cdot)$ denotes the matrix determinant, $n^{(s')}_z$ the dimensionality of the measurement space for sensor $s'$
and $n_x$ is the dimensionality of state space.
$M_J$, $b_J$ and $c_J$ are defined in Theorem~\ref{theorem::gm_psibar}.
\end{lemma}

\begin{proof}
Given that the single-sensor measurement likelihood and birth prior are Gaussian densities, the product can be expanded as,
\begin{multline*}
	\prod\limits_{\substack{s'=1\\j^{(s')} > 0}}^V g^{(s')}(z_{j^{(s')}}^{(s')} | x)p_B(x) =\\
		\left[\prod\limits_{\substack{s'=1\\j^{(s')} > 0}}^V (2\pi)^{-n^{(s')}_z/2}\det(R^{(s')})^{-1/2}\right.\\
		\times\left.\exp\left\{-\frac{1}{2} (z_{j^{(s')}}^{(s')} - H^{(s')}x)^TR^{(s'),-1}(z_{j^{(s')}}^{(s')} - H^{(s')}x)\right\}\right]\\
		\times\left[(2\pi)^{-n_x/2}\det(P_0)^{-1/2}\exp \left\{-\frac{1}{2}(x-\mu_0)^TP_0^{-1}(x-\mu_0)\right\}\right].\\
\end{multline*}
Rearranging the variables and leveraging the product properties of exponentials,
\begin{multline}\label{eq::lemma1::expanded}
	\prod\limits_{\substack{s'=1\\j^{(s')} > 0}}^V g^{(s')}(z_{j^{(s')}}^{(s')} | x)p_B(x) =\\
		\left[(2\pi)^{-n_x/2}\det(P_0)^{-1/2}\right]\left[\prod\limits_{\substack{s'=1\\j^{(s')} > 0}}^V (2\pi)^{-n^{(s')}_z/2}\det(R^{(s')})^{-1/2}\right]\\
		\times\exp \Bigg\{-\frac{1}{2}\Bigg(\sum\limits_{\substack{s'=1\\j^{(s')} > 0}}^V(z_{j^{(s')}}^{(s')} - H^{(s')}x)^TR^{(s'),-1}(z_{j^{(s')}}^{(s')} - H^{(s')}x)\\
		+ (x-\mu_0)^TP_0^{-1}(x-\mu_0)\Bigg)\Bigg\}.
\end{multline}
Let $\Psi$ be a temporary variable denoting the arguments in the exponential.
Factoring $\Psi$ into quadratic form,
\begin{multline*}
\Psi = x^TP_0^{-1}x - 2\mu_0^TP_0^{-1}x + \mu_0^TP^{-1}\mu_0\\
+\sum\limits_{\substack{s'=1\\j^{(s')} > 0}}^V  \Bigg(z_{j^{(s')}}^{(s'),T}R^{(s'),-1}z_{j^{(s')}}^{(s')} - 2z_{j^{(s')}}^{(s'),T}R^{(s'),-1}H^{(s')}x\\
+x^TH^{(s')}R^{(s'),-1}H^{(s')}x\Bigg).
\end{multline*}
Further, grouping like terms in $x$,
\begin{multline*}
\Psi = x^T\left(P_0^{-1} + \sum\limits_{\substack{s'=1\\j^{(s')} > 0}}^V H^{(s'),T}R^{(s'),-1}H^{(s')}\right)x\\
-2\left(P_0^{-1}\mu_0 + \sum\limits_{\substack{s'=1\\j^{(s')} > 0}}^V H^{(s'),T}R^{(s'),-1}z_{j^{(s')}}^{(s')}\right)^Tx\\
+\left(\mu_0^TP_0^{-1}\mu_0 + \sum\limits_{\substack{s'=1\\j^{(s')} > 0}}^V z_{j^{(s')}}^{(s'),T}R^{(s'),-1}z_{j^{(s')}}^{(s')}\right)
\end{multline*}
By substitution of variables $M_J$, $b_J$ and $c_J$ defined in Theorem~\ref{theorem::gm_psibar},
\begin{equation}\label{eq::gm::lemma1::exp_subform}
\Psi = x^TM_Jx - 2b_J^Tx + c_J
\end{equation}
Since $M_J$ is symmetric positive definite (Lemma~\ref{lemma::M_positive_definite}), we can complete the square of the first two terms in Equation~(\ref{eq::gm::lemma1::exp_subform}),
\begin{equation}\label{eq::gm::lemma1::finalpsi}
	\Psi = (x - M_J^{-1}b_J)^TM_J(x-M_J^{-1}b_J) - b_J^TM_J^{-1}b_J + c_J
\end{equation}
Finally, substituting Equation~(\ref{eq::gm::lemma1::finalpsi}) into Equation~(\ref{eq::lemma1::expanded}) and rearranging terms yields the final form of Lemma~\ref{lemma::gm1}.
\end{proof}

\begin{lemma}\label{lemma::gm2}
Under the linear Gaussian assumptions stated in Section~\ref{sec::gm},
\begin{multline*}
	\int\prod\limits_{\substack{s'=1\\j^{(s')} > 0}}^V g^{(s')}(z_{j^{(s')}}^{(s')} | x)p_B(x)dx =\\
	\left[\det(P_0)\det(M_J)\prod\limits_{\substack{s'=1\\j^{(s')} > 0}}^V (2\pi)^{n^{(s')}_z}\det(R^{(s')})\right]^{-\frac{1}{2}}\\
	\exp \left\{-\frac{1}{2} (c_J-b_J^T M_J^{-1} b_J) \right\}
\end{multline*}
\end{lemma}
\begin{proof}
Applying Lemma~\ref{lemma::gm1} to the integrand, the only term dependent on $x$ is the exponential and thus can be taken out of the integral.
Letting $y = x-M_J^{-1}b_J$ and $dy=dx$, the integral simplifies,
\begin{multline*}
\int \exp\left\{-\frac{1}{2}(x-M_J^{-1}b_J)^TM_J(x-M_J^{-1}b_J)\right\} dx = \\
\int \exp\left\{-\frac{1}{2}y^TM_Jy\right\}dy,
\end{multline*}
which under the standard Gaussian integral evaluates to,
\begin{equation}\label{eq::gm_int}
	\int \exp\left\{-\frac{1}{2}y^TM_Jy\right\}dy = \left(2\pi\right)^{n_x/2}\det(M_J)^{-1/2}
\end{equation}
By re-applying Lemma~\ref{lemma::gm1}, carrying over the constants and by using Equation~(\ref{eq::gm_int}) as the solution to the integral, we arrive at the 
final form of Lemma~\ref{lemma::gm2}.
\end{proof}

\begin{lemma}\label{lemma::I_in_truncation}
	Using the premises provided in Theorem~\ref{theorem::l1_dist},
	\begin{equation*}
		I_+ \cap (\mathbb{B}_+ \setminus \mathbb{B}'_+) \neq \emptyset, \hfill
		\forall I_+ \in (\mathcal{F}(\mathbb{L} \cup \mathbb{B}_+) \setminus \mathcal{F}(\mathbb{L} \cup \mathbb{B}'_+)).
	\end{equation*}
\end{lemma}

\begin{proof}
	Consider $I_+ \in (\mathcal{F}(\mathbb{L} \cup \mathbb{B}_+) \setminus \mathcal{F}(\mathbb{L} \cup \mathbb{B}'_+))$ and assume the opposite is true, $I_+ \cap (\mathbb{B}_+ \setminus \mathbb{B}'_+) = \emptyset$.
	Then $\forall l_+ \in I_+$, $l_+ \notin \mathbb{B}_+ \setminus \mathbb{B}'_+$ (or equivalently $l_+ \in \mathbb{L} \cup \mathbb{B}'_+$), resulting in $I_+ \subseteq \mathbb{L} \cup \mathbb{B}'_+$.
	However this contradicts the statement $I_+ \in (\mathcal{F}(\mathbb{L} \cup \mathbb{B}_+) \setminus \mathcal{F}(\mathbb{L} \cup \mathbb{B}'_+))$ since $I_+ \in \mathcal{F}(\mathbb{L} \cup \mathbb{B}'_+)$.
	Since the assumed premise leads to a contradiction, it follows that $I_+ \cap (\mathbb{B}_+ \setminus \mathbb{B}'_+) \neq \emptyset$.
\end{proof}

%% file: sections/appendix_l1_dist.tex
\section{Proof of Theorem~\ref{theorem::l1_dist}}\label{sec::appendix::proof_l1_dist}

\begin{proof}
    From \cite[proposition 5]{Vo2015}, the L1-distance between two $\delta$-\ac{GLMB} distributions is given as,
    \begin{equation*}
        |\pi_{\mathbb{H}} - \pi_{\mathbb{H}'}| = \sum\limits_{(I, \xi, I_+, \theta_+)\in \mathbb{H} \setminus \mathbb{H}'} w^{(I, \xi)}w_{Z_+}^{(I, \xi, I_+, \theta_+)}.
    \end{equation*}
    By the distributive property of cartesian products over set differences, $\mathbb{H} - \mathbb{H}'$ is,
    \begin{multline*}
        (\mathcal{F}(\mathbb{L}) \times \mathcal{F}(\mathbb{L} \cup \mathbb{B}_+) \times \Xi \times \Theta_+) \setminus
        (\mathcal{F}(\mathbb{L}) \times \mathcal{F}(\mathbb{L} \cup \mathbb{B}'_+) \times \Xi \times \Theta_+) =\\
        \mathcal{F}(\mathbb{L}) \times (\mathcal{F}(\mathbb{L} \cup \mathbb{B}_+) \setminus \mathcal{F}(\mathbb{L} \cup \mathbb{B}'_+)) \times \Xi \times \Theta_+.
    \end{multline*}
    By Lemma~\ref{lemma::I_in_truncation}, every $I_+ \in (\mathcal{F}(\mathbb{L} \cup \mathbb{B}_+) \setminus \mathcal{F}(\mathbb{L} \cup \mathbb{B}'_+))$ must have at least one label in the truncated set $\mathbb{T}_+ = \mathbb{B}_+ \setminus \mathbb{B}'_+$.
    Equation (\ref{eq::glmb_update::w}) can be rearranged as,
    \begin{multline*}
        w_{Z_+}^{(I, \xi, I_+, \theta_+)} = 
        1_{\Theta_+(I_+)}(\theta_+)
        \left[1 - \bar{P}^{(\xi)}_s \right]^{I-I_+}
        \left[ \bar{P}^{(\xi)}_s\right]^{I \cap I_+}\\
        \left[1 - r_{B,+}\right]^{\mathbb{B}_+ - I_+}
        r_{B, +}^{(\mathbb{B}_+\cap I_+) - \mathbb{T}_+}\\
        \left[\bar{\psi}^{(\xi, \theta_+)}_{Z_+}\right]^{I_+}
        \left[r_{B, +}\right]^{I_+ \cap \mathbb{T}_+},
    \end{multline*}
    to clearly delineate the contribution by labels in the truncated set $\mathbb{T}_+$.
    Note that all of the multiplicands are upper bounded by 1 except $\left[\bar{\psi}^{(\xi, \theta_+)}_{Z_+}\right]^{I_+}$ and $\left[r_{B, +}\right]^{I_+ \cap \mathbb{T}_+}$.
    By construction of the truncated set $\mathbb{T}_+$, it holds that $r_{B, +}(l_+) < \epsilon$ for $l_+ \in \mathbb{T}_+$.
    If $K$ is a positive upper bound, $0 \leq \bar{\psi}^{(\xi, \theta_+)}_{Z_+}(l_+) \leq K, \; \forall I_+ \in \mathcal{F}(\mathbb{L} \cup (\mathbb{B}_+ \setminus \mathbb{B}'_+))$ with $l_+ \in I_+$ and $\theta_+ \in \Theta_+(I_+)$,
    the product is bounded by,
    \begin{equation}\label{eq::w_bound}
        w_{Z_+}^{(I, \xi, I_+, \theta_+)} \leq K^{|I_+|} \epsilon^{N_{\mathbb{T}_+}(I_+)},
    \end{equation}
    where $N_{\mathbb{T}_+}(I_+) = | I_+ \cap (\mathbb{B}_+ \setminus \mathbb{B}'_+)|$ is the number of truncated labels in $I_+$.
    By bounding the normalized prior $w^{(\xi, I)} \leq 1$ and Equation (\ref{eq::w_bound}) into the L1-distance equation \cite[proposition 5]{Vo2014} we reach final expression.
\end{proof}

%% file: sections/appendix_gaussian_expval.tex
\section{Proof of Theorem~\ref{theorem::gm_psibar}}\label{sec::appendix::proof_gm_theorem0}

\begin{proof}
Grouping missed detected and detected elements of $J$ in $\psi^J_Z$ results in,
\begin{equation}\label{eq::psiJ_expand}
    \psi^J_Z(x,l_+) = 
        \left[\prod\limits^V_{\substack{s'=1\\j^{(s') = 0}}} (1-p^{(s')}_D)\right]
        \left[\prod\limits^V_{\substack{s'=1\\j^{(s') > 0}}}\frac{p^{(s')}_D g^{(s')}(z^{(s')}_{j^{(s')}}|x)}{\kappa^{(s)}(z^{(s')}_{j^{(s')}})}\right].
\end{equation}
By substituting into Equation~(\ref{eq::psi_bar}) and removing elements from the integrand that are not a function of $x$,
\begin{multline*}
    \bar{\psi}^J_Z(l_+) = 
    \left[\prod\limits^V_{\substack{s'=1\\j^{(s') = 0}}} (1-p^{(s')}_D)\right]
    \left[\prod\limits^V_{\substack{s'=1\\j^{(s') > 0}}}\frac{p^{(s')}_D}{\kappa^{(s)}(z^{(s')}_{j^{(s')}})}\right]\\
    \int \prod\limits^V_{\substack{s'=1\\j^{(s') > 0}}}g^{(s')}(z^{(s')}_{j^{(s')}}|x) p_{B}(x)dx.
\end{multline*}
Applying Lemma~\ref{lemma::gm2} results in the final expression.
\end{proof}

%% file: sections/appendix_gaussian_sampling.tex
\section{Proof of Theorem~\ref{theorem::gm_sample_distr}}\label{sec::appendix::proof_gm_theorem1}

\begin{proof}
If $j^{(s)}>0$, then by substitution of Equation~(\ref{eq::gauss_psibar}) into Equation~(\ref{eq::cdn_likelihood}) and combining values that are not a function of sensor $s$ into the normalizing constant results in the final expression when $j^{(s)} > 0$ in Theorem~\ref{theorem::gm_sample_distr}.

If $j^{(s)}=0$, Equation~(\ref{eq::cdn_likelihood}) reduces as follows,
\begin{equation*}
	p(j^{(s)} | J^{-s}) \propto \left(1-p^{(s)}_D\right)\bar{\psi}^{J^{-s}}_Z(l_+).
\end{equation*}
Since, $\left(1-p^{(s)}_D\right)$ is not a function of $x$ and by notation, $\left(1 - r_U(z^{(s)}_{j^{(s)}})\right) = 1$ when $j^{(s)} = 0$.
Substituting the expression from Equation~(\ref{eq::gauss_psibar}) over the domain $J^{-s}$, and combining values that are not a function of sensor $s$ into the normalizing constant results in the final expression for $j^{(s)} = 0$ in Theorem~\ref{theorem::gm_sample_distr}.
Note that, although $M_{J^{-s}}$ and $\psi_{J^{-s}}$ are not a function of sensor $s$, they cannot be combined into the normalizing constant as it is not easy to separate these terms within the determinant or inverse for the $j^{(s)} > 0$ case.
Similarly, although $c_J$ can be factored as,
\begin{equation*}
	c_J = c_{J^{-s}} + 1_{\mathbb{J}}(j^{(s)}) z_{j^{(s)}}^{(s)T}R^{(s),-1}z_{j^{(s)}}^{(s)},
\end{equation*}
where,
\begin{equation*}
	c_{J^{-s}} =  \mu_0^TP^{-1}_0\mu_0 + \sum\limits_{s' \in J^{-s}}z_{j^{(s')}}^{(s'),T}R^{(s'),-1}z_{j^{(s')}}^{(s')},
\end{equation*}
in practice, by not including $c_{J^{-s}}$ in $\Phi_J$, it may result in numerical instabilities within the exponential function and thus not recommended to be combined into the normalizing constant.
\end{proof}

%% file: sections/appendix_gaussian_spatial.tex
\section{Proof of Theorem~\ref{theorem::gm_spatial_distr}}\label{sec::appendix::proof_gm_spatial_distr}

\begin{proof}
By Equation~(\ref{eq::psiJ_expand}), the only term that is a factor of $x$ is $g^{(s')}(z^{(s')}_{j^{(s')}}|x)$.
Because of this, all terms except,
	$\prod\limits^V_{\substack{s=1\\j^{(s')} > 0}}g^{(s')}(z^{(s')}_{j^{(s')}}|x)$
can be factored out of the denominator of Equation~(\ref{eq::spatial_distr}) and cancels with the terms in the numerator resulting in,
\begin{equation*}
	p_B \left( x, l_+ | Z_J\right) = \frac{p_B(x, l) \prod\limits^V_{\substack{s=1\\j^{(s')} > 0}}g^{(s')}(z^{(s')}_{j^{(s')}}|x)}{\bigg\langle p_B, \prod\limits^V_{\substack{s=1\\j^{(s')} > 0}}g^{(s')}(z^{(s')}_{j^{(s')}}|\cdot)\bigg\rangle}
\end{equation*}
Using Lemmas~\ref{lemma::gm1} and~\ref{lemma::gm2} for the numerator and denominator respectively, the posterior can be simplified as,
\begin{multline*}
	p_B \left( x, l_+ | Z_J\right) = (2\pi)^{-n_x/2}\det(M_J)^{1/2}\\
	\times \exp \left\{ -\frac{1}{2} (x - M_J^{-1}b_J)^TM_J(x-M_J^{-1}b_J)\right\}.
\end{multline*}
Letting $\mu' = M_J^{-1}b_J$ and $P'=M_J^{-1}$ results in a Gaussian of the form $p_B \left( x, l_+ | Z_J\right) = \mathcal{N}(x; \mu', P')$.
Finally, for linear Gaussian system, the posterior density predicted to the next time step is given by \cite{Vo2015},
\begin{equation*}
	p_{B,+} \left( x, l_+ | Z_J\right) = \mathcal{N}(x; F\mu', FP'F^T + Q)
\end{equation*}
which is the final Gaussian density expression in Theorem~\ref{theorem::gm_spatial_distr}.
\end{proof}

%% file: ms.bbl
\begin{thebibliography}{10}
\providecommand{\url}[1]{#1}
\csname url@samestyle\endcsname
\providecommand{\newblock}{\relax}
\providecommand{\bibinfo}[2]{#2}
\providecommand{\BIBentrySTDinterwordspacing}{\spaceskip=0pt\relax}
\providecommand{\BIBentryALTinterwordstretchfactor}{4}
\providecommand{\BIBentryALTinterwordspacing}{\spaceskip=\fontdimen2\font plus
\BIBentryALTinterwordstretchfactor\fontdimen3\font minus
  \fontdimen4\font\relax}
\providecommand{\BIBforeignlanguage}[2]{{%
\expandafter\ifx\csname l@#1\endcsname\relax
\typeout{** WARNING: IEEEtran.bst: No hyphenation pattern has been}%
\typeout{** loaded for the language `#1'. Using the pattern for}%
\typeout{** the default language instead.}%
\else
\language=\csname l@#1\endcsname
\fi
#2}}
\providecommand{\BIBdecl}{\relax}
\BIBdecl

\bibitem{Blackman1999}
S.~Blackman and R.~Popoli, \emph{Design and Analysis of Modern Tracking
  Syst.}\hskip 1em plus 0.5em minus 0.4em\relax Norwood, MA, USA: Artech House,
  1999.

\bibitem{Barshalom2009}
Y.~Bar-Shalom, F.~Daum, and J.~Huang, ``The probabilistic data association
  filter,'' \emph{IEEE Control Syst. Mag.}, vol.~29, no.~6, pp. 82--100, 2009.

\bibitem{Blackman2004}
S.~Blackman, ``Multiple hypothesis tracking for multiple target tracking,''
  \emph{IEEE Aerosp. Electron. Syst. Mag.}, vol.~19, no.~1, pp. 5--18, Jan
  2004.

\bibitem{Meyer2018}
F.~Meyer, T.~Kropfreiter, J.~L. Williams, R.~Lau, F.~Hlawatsch, P.~Braca, and
  M.~Z. Win, ``Message passing algorithms for scalable multitarget tracking,''
  \emph{Proc. IEEE}, vol. 106, no.~2, pp. 221--259, Feb 2018.

\bibitem{Mahler2007}
R.~Mahler, \emph{Statistical Multisource-Multitarget Information Fusion}.\hskip
  1em plus 0.5em minus 0.4em\relax Norwood, MA, USA: Artech House, 2007.

\bibitem{Mahler2014}
------, \emph{Advances in Statistical Multisource-Multitarget Information
  Fusion}.\hskip 1em plus 0.5em minus 0.4em\relax Norwood, MA, USA: Artech
  House, 2014.

\bibitem{Vo2015}
B.-N. Vo, M.~Mallick, Y.~Bar-Shalom, S.~Coraluppi, R.~Mahler, and B.-T. Vo,
  ``Multitarget tracking,'' \emph{Wiley Encyclopedia of Elect. and Electron.
  Eng.}, Sep. 2015.

\bibitem{Kennedy2008}
H.~L. Kennedy, ``Comparison of {MHT} and {PDA} track initiation performance,''
  in \emph{Proc. IEEE Int. Conf. Radar}, 2008, pp. 508--512.

\bibitem{Hu1997}
Z.~Hu, H.~Leung, and M.~Blanchette, ``Statistical performance analysis of track
  initiation techniques,'' \emph{IEEE Trans. Signal Process.}, vol.~45, no.~2,
  pp. 445--456, 1997.

\bibitem{Vo2013}
B.-T. Vo and B.-N. Vo, ``Labeled random finite sets and multi-object conjugate
  priors,'' \emph{IEEE Trans. Signal Process.}, vol.~61, no.~13, pp.
  3460--3475, 2013.

\bibitem{Vo2006}
B.-N. Vo and W.-K. Ma, ``The {G}aussian mixture probability hypothesis density
  filter,'' \emph{IEEE Trans. Signal Process.}, vol.~54, no.~11, p. 4091, 2006.

\bibitem{Vo2007}
B.-T. Vo, B.-N. Vo, and A.~Cantoni, ``Analytic implementations of the
  cardinalized probability hypothesis density filter,'' \emph{IEEE Trans.
  Signal Process.}, vol.~55, no.~7, pp. 3553--3567, 2007.

\bibitem{Garcia2018}
{\'A}.~F. Garc{\'\i}a-Fern{\'a}ndez, J.~L. Williams, K.~Granstr{\"o}m, and
  L.~Svensson, ``Poisson multi-{B}ernoulli mixture filter: Direct derivation
  and implementation,'' \emph{IEEE Trans. Aerosp. Electron. Syst.}, vol.~54,
  no.~4, pp. 1883--1901, 2018.

\bibitem{Reuter2014}
S.~Reuter, B.-T. Vo, B.-N. Vo, and K.~Dietmayer, ``The labeled
  multi-{B}ernoulli filter,'' \emph{IEEE Trans. Signal Process.}, vol.~62,
  no.~12, pp. 3246--3260, 2014.

\bibitem{Reuter2017}
S.~Reuter, A.~Danzer, M.~St{\"u}bler, A.~Scheel, and K.~Granstr{\"o}m, ``A fast
  implementation of the labeled multi-{B}ernoulli filter using {G}ibbs
  sampling,'' in \emph{Proc. IEEE Intell. Veh. Symp.}, 2017, pp. 765--772.

\bibitem{Vo2014}
B.-N. Vo, B.-T. Vo, and D.~Phung, ``Labeled random finite sets and the {B}ayes
  multi-target tracking filter,'' \emph{IEEE Trans. Signal Process.}, vol.~62,
  no.~24, pp. 6554--6567, 2014.

\bibitem{Vo2017}
B.-N. Vo, B.-T. Vo, and H.~G. Hoang, ``An efficient implementation of the
  generalized labeled multi-{B}ernoulli filter,'' \emph{IEEE Trans. Signal
  Process.}, vol.~65, no.~8, pp. 1975--1987, 2016.

\bibitem{Vo2019}
B.-N. Vo, B.-T. Vo, and M.~Beard, ``Multi-sensor multi-object tracking with the
  generalized labeled multi-{B}ernoulli filter,'' \emph{IEEE Trans. Signal
  Process.}, vol.~67, no.~23, pp. 5952--5967, 2019.

\bibitem{Skolnik2008}
M.~Skolnik, \emph{Radar handbook}.\hskip 1em plus 0.5em minus 0.4em\relax New
  York, NY, USA: McGraw-Hill Education, 2008.

\bibitem{Ristic2012}
B.~Ristic, D.~Clark, B.-N. Vo, and B.-T. Vo, ``Adaptive target birth intensity
  for {PHD} and {CPHD} filters,'' \emph{IEEE Trans. Aerosp. Electron. Syst.},
  vol.~48, no.~2, pp. 1656--1668, 2012.

\bibitem{Beard2013}
M.~Beard, B.-T. Vo, B.-N. Vo, and S.~Arulampalam, ``A partially uniform target
  birth model for {G}aussian mixture {PHD}/{CPHD} filtering,'' \emph{IEEE
  Trans. Aerosp. Electron. Syst.}, vol.~49, no.~4, pp. 2835--2844, 2013.

\bibitem{Reuter2013}
S.~Reuter, D.~Meissner, B.~Wilking, and K.~Dietmayer, ``Cardinality balanced
  multi-target multi-{B}ernoulli filtering using adaptive birth
  distributions,'' in \emph{Proc. IEEE 16th Int. Conf. Inf. Fusion}, 2013, pp.
  1608--1615.

\bibitem{Changshun2018}
Y.~Changshun, W.~Jun, L.~Peng, and S.~Jinping, ``Adaptive multi-{B}ernoulli
  filter without need of prior birth multi-{B}ernoulli random finite set,''
  \emph{Chinese Journal Electron.}, vol.~27, no.~1, pp. 115--122, 2018.

\bibitem{Hu2019}
X.~Hu, H.~Ji, and L.~Liu, ``Adaptive target birth intensity multi-{B}ernoulli
  filter with noise-based threshold,'' \emph{Sensors}, vol.~19, no.~5, p. 1120,
  2019.

\bibitem{Lin2016}
S.~Lin, B.-T. Vo, and S.~E. Nordholm, ``Measurement driven birth model for the
  generalized labeled multi-{B}ernoulli filter,'' in \emph{Proc. Int. Conf.
  Control, Automat. and Inf. Sci.}, 2016, pp. 94--99.

\bibitem{Legrand2018}
K.~A. LeGrand and K.~J. DeMars, ``The data-driven delta-generalized labeled
  multi-{B}ernoulli tracker for automatic birth initialization,'' in
  \emph{Proc. SPIE}, vol. 10646, 2018, p. 1064606.

\bibitem{Zhu2021}
S.~Zhu, B.~Yang, and S.~Wu, ``Measurement-driven multi-target tracking filter
  under the framework of labeled random finite set,'' \emph{Digital Signal
  Process.}, vol. 112, p. 103000, 2021.

\bibitem{Hoher2020}
P.~Hoher, T.~Baur, S.~Wirtensohn, and J.~Reuter, ``A detection driven adaptive
  birth density for the labeled multi-{B}ernoulli filter,'' in \emph{Proc. IEEE
  23rd Int. Conf. Inf. Fusion}, 2020, pp. 1--8.

\bibitem{Do2020}
C.-T. Do, T.~T.~D. Nguyen, and D.~Moratuwage, ``Multi-target tracking with an
  adaptive $\delta-${GLMB} filter,'' \emph{arXiv preprint arXiv:2008.00413},
  2020.

\bibitem{Yoon2011}
J.~H. Yoon, D.~Y. Kim, S.~H. Bae, and V.~Shin, ``Joint initialization and
  tracking of multiple moving objects using doppler information,'' \emph{IEEE
  Trans. Signal Process.}, vol.~59, no.~7, pp. 3447--3452, 2011.

\bibitem{Liu2020}
Z.-X. Liu, J.~Gan, J.-S. Li, and M.~Wu, ``Adaptive $\delta$-generalized labeled
  multi-{B}ernoulli filter for multi-object detection and tracking,''
  \emph{IEEE Access}, vol.~9, pp. 2100--2109, 2020.

\bibitem{Gostar2021}
A.~K. Gostar, T.~Rathnayake, R.~Tennakoon, A.~Bab-Hadiashar, G.~Battistelli,
  L.~Chisci, and R.~Hoseinnezhad, ``Centralized cooperative sensor fusion for
  dynamic sensor network with limited field-of-view via labeled
  multi-{B}ernoulli filter,'' \emph{IEEE Trans. Signal Process.}, vol.~69, pp.
  878--891, 2020.

\bibitem{Tian2019}
Z.~Tian, W.~Liu, and X.~Ru, ``Multi-acoustic array localization and tracking
  method based on {G}ibbs-{GLMB},'' in \emph{Proc. Int. Conf. Control, Automat.
  and Inf. Sci.}, 2019, pp. 1--6.

\bibitem{Frohle2019}
M.~Fr{\"o}hle, C.~Lindberg, K.~Granstr{\"o}m, and H.~Wymeersch, ``Multisensor
  poisson multi-bernoulli filter for joint target-sensor state tracking,''
  \emph{IEEE Trans. Intell. Veh.}, vol.~4, no.~4, pp. 609--621, 2019.

\bibitem{Lanterman2008}
M.~Tobias and A.~Lanterman, ``Techniques for birth-particle placement in the
  probability hypothesis density particle filter applied to passive radar,''
  \emph{IET Radar, Sonar \& Navigation}, vol.~2, no.~5, pp. 351--365, 2008.

\bibitem{Berry2018}
C.~Berry, D.~J. Bucci, and S.~W. Schmidt, ``Passive multi-target tracking using
  the adaptive birth intensity {PHD} filter,'' in \emph{Proc. IEEE 21st Int.
  Conf. Inf. Fusion}, 2018, pp. 353--360.

\bibitem{Liu2017}
W.~Liu, Y.~Chen, H.~Cui, and Q.~Ge, ``Multi-sensor tracking with
  non-overlapping field for the {GLMB} filter,'' in \emph{Proc. Int. Conf.
  Control, Automat. and Inf. Sci.}, 2017, pp. 197--202.

\bibitem{Li2019}
S.~Li, G.~Battistelli, L.~Chisci, W.~Yi, B.~Wang, and L.~Kong,
  ``Computationally efficient multi-agent multi-object tracking with labeled
  random finite sets,'' \emph{IEEE Trans. Signal Process.}, vol.~67, no.~1, pp.
  260--275, 2019.

\bibitem{Li2018}
S.~Li, W.~Yi, R.~Hoseinnezhad, G.~Battistelli, B.~Wang, and L.~Kong, ``Robust
  distributed fusion with labeled random finite sets,'' \emph{IEEE Trans.
  Signal Process.}, vol.~66, no.~2, pp. 278--293, 2018.

\bibitem{Bishop2007}
A.~N. Bishop and P.~N. Pathirana, ``Localization of emitters via the
  intersection of bearing lines: A ghost elimination approach,'' \emph{IEEE
  Trans. Veh. Technol.}, vol.~56, no.~5, pp. 3106--3110, 2007.

\bibitem{Jia2017}
T.~Jia, H.~Wang, X.~Shen, X.~Liu, and H.~Jing, ``Bearing-only multiple sources
  localization and the spatial spectrum,'' in \emph{Proc. IEEE OCEANS},
  Aberdeen, Scotland, 2017, pp. 1--5.

\bibitem{Reed2008}
J.~D. Reed, C.~R. da~Silva, and R.~M. Buehrer, ``Multiple-source localization
  using line-of-bearing measurements: Approaches to the data association
  problem,'' in \emph{Proc. IEEE Military Commun. Conf.}, 2008, pp. 1--7.

\bibitem{Alexandridis2015}
A.~Alexandridis, G.~Borboudakis, and A.~Mouchtaris, ``Addressing the
  data-association problem for multiple sound source localization using {DOA}
  estimates,'' in \emph{Proc. IEEE 23rd Eur. Signal Process. Conf.}, 2015, pp.
  1551--1555.

\bibitem{Wang2020}
L.~Wang, Y.~Yang, and X.~Liu, ``A direct position determination approach for
  underwater acoustic sensor networks,'' \emph{IEEE Trans. Veh. Technol.},
  vol.~69, no.~11, pp. 13\,033--13\,044, 2020.

\bibitem{Papi2016}
F.~Papi, ``Multi-sensor $\delta$-{GLMB} filter for multi-target tracking using
  doppler only measurements,'' in \emph{Proc. IEEE Eur. Intell. and Secur.
  Inform. Conf.}, 2015, pp. 83--89.

\bibitem{Roberts1994}
G.~O. Roberts and A.~F. Smith, ``Simple conditions for the convergence of the
  {G}ibbs sampler and {M}etropolis-{H}astings algorithms,'' \emph{Stochastic
  Processes and their Applications}, vol.~49, no.~2, pp. 207--216, 1994.

\bibitem{Geyer1995}
C.~J. Geyer and E.~A. Thompson, ``Annealing {M}arkov chain {M}onte {C}arlo with
  applications to ancestral inference,'' \emph{Journal of the Amer. Statistical
  Assoc.}, vol.~90, no. 431, pp. 909--920, 1995.

\bibitem{Neal2001}
R.~M. Neal, ``Annealed importance sampling,'' \emph{Statist. and Comput.},
  vol.~11, no.~2, pp. 125--139, 2001.

\bibitem{Bishop2006}
C.~Bishop, \emph{Pattern Recognition and Machine Learning}.\hskip 1em plus
  0.5em minus 0.4em\relax New York, NY USA: Springer, 2006.

\bibitem{Li2015}
T.~Li, M.~Bolic, and P.~M. Djuric, ``Resampling methods for particle filtering:
  {C}lassification, implementation, and strategies,'' \emph{IEEE Signal Proc.
  Mag.}, vol.~32, no.~3, pp. 70--86, 2015.

\bibitem{Crassidis2011}
J.~L. Crassidis and J.~L. Junkins, \emph{Optimal Estimation of Dynamic
  Systems}.\hskip 1em plus 0.5em minus 0.4em\relax CRC press, 2011.

\bibitem{Li2003}
X.~R. Li and V.~P. Jilkov, ``Survey of maneuvering target tracking. part {I} :
  Dynamic models,'' \emph{IEEE Trans. Aerosp. Electron. Syst.}, vol.~39, no.~4,
  pp. 1333--1364, 2003.

\bibitem{Beard2020}
M.~Beard, B.-T. Vo, and B.-N. Vo, ``A solution for large-scale multi-object
  tracking,'' \emph{IEEE Trans. Signal Process.}, vol.~68, pp. 2754--2769,
  2020.

\bibitem{Wolf2020}
L.~M. Wolf and M.~Baum, ``Deterministic {G}ibbs sampling for data association
  in multi-object tracking,'' in \emph{Proc. IEEE Int. Conf. Multisensor Fusion
  and Integration for Intell. Syst.}, 2020, pp. 291--296.

\bibitem{Horn2012}
R.~A. Horn and C.~R. Johnson, \emph{Matrix Analysis}, 2nd~ed.\hskip 1em plus
  0.5em minus 0.4em\relax USA: Cambridge University Press, 2012.

\end{thebibliography}
